\documentclass[a4paper]{article}
\usepackage[T1]{fontenc}
\usepackage[utf8]{inputenc}
\usepackage{lmodern}
\usepackage{geometry}
\usepackage{amsmath}
\usepackage{amsfonts}
\usepackage{amssymb}
\usepackage{amsthm}
\usepackage{mathtools}
\usepackage{commath}
\usepackage{braket}
\usepackage{bm}
\usepackage[all]{xy}
\usepackage{hyperref}

\theoremstyle{definition}
\newtheorem{defn}{Definition}
\theoremstyle{plain}
\newtheorem{teo}[defn]{Theorem}
\newtheorem{lem}[defn]{Lemma}
\numberwithin{equation}{section}

\DeclareMathOperator{\Der}{Der}
\DeclareMathOperator{\diag}{diag}

\DeclareMathOperator{\GL}{GL}

\DeclareMathOperator{\id}{id}
\DeclareMathOperator{\Mat}{Mat}
\DeclareMathOperator{\PGL}{PGL}
\DeclareMathOperator{\Rep}{Rep}
\DeclareMathOperator{\spa}{span}
\DeclareMathOperator{\Spec}{Spec}
\DeclareMathOperator{\tr}{tr}

\mathchardef\mhyphen="2D
\newcommand{\alg}[1]{{#1\mhyphen\mathbf{Alg}}}
\newcommand{\app}[3]{#1\colon #2\to #3}
\newcommand{\bb}[1]{\mathbb{#1}}
\newcommand{\bimod}[1]{{#1\mhyphen\mathbf{Bimod}}}
\newcommand{\comm}[2]{[#1,#2]}
\newcommand{\cq}[2]{#1\slash\!\!\slash #2}
\newcommand{\de}{\mathrm{d}}
\newcommand{\deq}{\mathrel{:=}}
\newcommand{\dgalg}[1]{{#1\mhyphen\mathbf{dga}}}
\newcommand{\DR}{\mathrm{DR}}
\newcommand{\falg}[2]{#1\langle #2\rangle}
\newcommand{\gc}[2]{[\![#1,#2]\!]}
\newcommand{\imm}{\mathrm{i}}
\newcommand{\isom}{\simeq}
\newcommand{\la}[1]{\mathfrak{#1}}
\newcommand{\lmod}[1]{{#1\mhyphen\mathbf{Mod}}}
\newcommand{\mph}{\makebox[0.7em]{$\cdot$}}
\newcommand{\nc}{\mathrm{nc}}
\newcommand{\ol}[1]{\overline{#1}}
\newcommand{\op}{\mathrm{op}}
\newcommand{\p}[1]{(#1)}
\newcommand{\pair}[2]{\langle #1,#2\rangle}
\newcommand{\poibr}[2]{\{#1,#2\}}
\newcommand{\rist}[2]{\left. #1\right|_{#2}}

\newcommand{\ind}[1]{\tilde{#1}}
\newcommand{\wind}[1]{\widetilde{#1}}
\newcommand{\tind}[1]{\hat{#1}}
\newcommand{\wtind}[1]{\widehat{#1}}

\begin{document}

\title{An introduction to associative geometry\\
  with applications to integrable systems}
\author{Alberto Tacchella}
\maketitle

\begin{abstract}
  The aim of these notes is to provide a reasonably short and ``hands-on''
  introduction to the differential calculus on associative algebras over a
  field of characteristic zero. Following a suggestion of Ginzburg's we call
  the resulting theory \emph{associative geometry}. We argue that this
  formalism sheds a new light on some classic solution methods in the theory
  of finite-dimensional integrable dynamical systems.
\end{abstract}

\tableofcontents
\clearpage

\section{Introduction}

The fundamental relationship between algebra and geometry is well known since
the times of Fermat and Descartes. In the modern language of category theory
this relationship is expressed in terms of equivalences of the form
\begin{equation*}
  \xymatrix@C=3pc{
    \mathbf{Spa} \ar@<0.5ex>[r]^-{\mathcal{O}} & \mathbf{Alg}^{\op}
    \ar@<0.5ex>[l]^-{\Spec}}
\end{equation*}
where \(\mathbf{Spa}\) is some category of \emph{spaces} and \(\mathbf{Alg}\)
is some category of (commutative) \emph{algebras}. The functor \(\mathcal{O}\)
maps a space to the algebra of ``regular'' (in the appropriate sense)
functions on it; the functor \(\Spec\) maps an algebra to its \emph{spectrum},
an object of the category \(\mathbf{Spa}\) naturally associated with it. Two
celebrated instances of this kind of construction are \emph{Gelfand duality},
relating compact Hausdorff spaces to commutative C*-algebras, and the basic
duality of modern algebraic geometry, relating affine schemes to commutative
rings (see \cite[Chapter 1]{khal09} for a detailed exposition of these and
other examples of this kind).

In this framework it is natural to ask whether this picture can be broadened
by taking as \(\mathbf{Alg}\) some category of associative \emph{but not
  necessarily commutative} algebras. It became clear early on that a naive
approach to this question is not viable: namely, one cannot simply extend in
any non-trivial way the usual definition of spectrum coming from algebraic
geometry to the category of all rings. This insight has been recently
formalized as a set of actual no-go theorems \cite{reyes12,vdbh14}.

To cope with this problem many different strategies for doing geometry on
particular classes of noncommutative algebras have been developed. In the
approach pioneered by Alain Connes in the 1980s and popularized in the book
\cite{conn94}, the idea is to use as a starting point the dictionary provided
by the aforementioned Gelfand duality and interpret the theory of (not
necessarily commutative) C*-algebras as a kind of ``noncommutative topology''.
When needed, this picture can be further refined by introducing analogues of
smooth structures and Riemannian metrics. This gave rise to a very rich theory
which is deeply rooted in functional analysis. We refer again the reader to
\cite{khal09} for a recent and very readable introduction to this field.

Another possible strategy to pursue is to generalize the usual
algebro-geometric concepts to some (hopefully large) class of noncommutative
rings. Here we encounter another important distinction, which corresponds to
the classical split between projective and affine algebraic geometry. In the
projective case, one is naturally led to study suitable classes of
\emph{graded} noncommutative rings. This is the approach taken, for instance,
by Artin and Zhang in their seminal paper \cite{az94}. The resulting theory,
which is known as ``noncommutative projective geometry'', is beautifully
described in the surveys \cite{svdb01}, \cite{staf02} and \cite{keel01}.

The affine case can be further divided, following \cite[Section 1]{ginzlect},
in two main strands. The first, that Ginzburg calls ``noncommutative geometry
in the small'', is best seen as a \emph{generalization} of conventional
(affine) algebraic geometry. Here one is typically interested in some sort of
noncommutative deformation of commutative algebras like superalgebras, rings
of differential operators and universal enveloping algebras of Lie algebras.
This kind of investigations is strictly related to the various mathematical
approaches to the problem of \emph{quantization}.

On the other hand, the second approach (called ``noncommutative geometry in
the large'' by Ginzburg) is a completely new theory that does \emph{not}
reduce to the commutative one as a special case. In this approach one deals
with \emph{generic} associative algebras, the basic examples being given by
free ones (namely, algebras of noncommutative polynomials in a finite number
of variables). In these notes we are going to explain in more detail this
point of view; following a suggestion of Ginzburg's, we shall refer to this
approach by the name of \emph{associative geometry}.

To recover some degree of geometric intuition in this very general setting the
following perspective, usually attributed to Kontsevich (see \cite[Section 9]
{kont93}), is very helpful. Let \(\bb{K}\) be a field of characteristic zero.
For each \(d\in \bb{N}\) we have a \emph{representation functor}
\begin{equation*}
  \app{\Rep_{d}}{\mathbf{AsAlg}}{\mathbf{AffSch}}
\end{equation*}
mapping each associative \(\bb{K}\)-algebra \(A\) to the affine scheme of
\(d\)-dimensional representations of \(A\) (that is, algebra morphisms \(A\to
\Mat_{d,d}(\bb{K})\)). According to Kontsevich, the ``associative-geometric''
objects on \(A\) are precisely those objects which induce in a natural way a
family of the corresponding (commutative) objects on each scheme in the
sequence \((\Rep_{d}(A))_{d\in \bb{N}}\). In other words, one can see
associative-geometric objects on \(A\) as ``blueprints'' for an infinite
sequence of ordinary geometric objects defined on representation spaces of
\(A\), each scheme \(\Rep_{d}(A)\) giving an increasingly better approximation
to the mysterious geometry determined by \(A\).

The first aim of these notes is to provide a reasonably compact survey of the
fundamental constructions and results at the basis of this circle of ideas.
The second aim is to show how the resulting theory provides a new
interpretation for some classic solution techniques in the field of
finite-dimensional integrable dynamical systems.

In more detail, the paper is organized as follows. In section \ref{s:ncdf} we
review the definition of the basic notions of differential calculus on a
generic associative algebra over a field of characteristic zero. In particular
in \S\ref{assoc-df} we build the \emph{Karoubi-de Rham complex}, whose
elements play the role of ``associative differential forms''. We analyze in
detail a couple of examples, the associative affine spaces (\S\ref{nc-aff-sp})
and the associative varieties determined by path algebras of quivers
(\S\ref{qpalg}). To deal with the latter it will be necessary to slightly
refine the class of associative differential forms used by introducing a
notion of differential calculus relative to a subalgebra (\S\ref{rel-df}).

In section \ref{s:rep} we review the connection between the worlds of
associative and commutative geometry. Following Kontsevich's philosophy
recalled above, one is led to consider the space of finite-dimensional
representations of a (finitely generated) associative algebra. This can be
interpreted as an affine scheme (or variety) on which a natural action of the
general linear group is defined. We explain in some detail the basic process
through which associative-geometric objects defined on the algebra \(A\)
induce \(\GL\)-invariant geometric objects on representation spaces. The
relative case, which is the appropriate one for quiver representation spaces,
is treated in \S\ref{qrep}.

Finally in section \ref{s:appl} we survey the applications of this formalism
to finite-dimensional integrable dynamical systems. We first review the
development by Kontsevich and Ginzburg of the associative version of
\emph{symplectic varieties} (inducing ordinary symplectic structures on
representation spaces) and the definition of the canonical associative
symplectic structures on free algebras and quiver path algebras. Then we
consider some simple examples of Hamiltonian systems on associative spaces,
and show how their (trivial) solutions project down to interesting flows on
some symplectic quotients of the corresponding representation spaces. This
approach can be seen as a natural extension of the \emph{projection method}
introduced by Olshanetsky and Perelomov to solve the systems of Calogero-Moser
type \cite{op81}.

In order to keep our treatment within reasonable bounds we were forced to
gloss over some more recent developments in associative geometry such as the
introduction of \emph{bisymplectic structures} \cite{cbeg07} and \emph{double
  Poisson structures} \cite{vdb08}. We hope to be able to cover these
important topics (and their applications to integrable systems) in a sequel to
these notes \cite{nc2}.

As it should be clear from the above summary, this paper is meant to be purely
expository. Every construction we are going to review can be found in more
advanced sources such as \cite{ginzlect} and \cite{lebr08}. On the reader's
part we assume a basic acquaintance with the fundamental notions of algebraic
(or differential) geometry, but little or no experience in dealing with
noncommutative algebras. We also assume a reasonable amount of familiarity
with basic category theory (especially universal constructions and adjoint
functors), and (for the material in section 4) a working knowledge of ordinary
symplectic geometry.

It should be stressed that I am not an expert in noncommutative geometry; the
content of these notes merely reflects my understanding of a small part of
this topic at the time of the deadline for submitting the manuscript. I hope
this effort will be useful for other novices who intend to venture into this
complex and fascinating field.

\paragraph{Acknowledgments.} It is a pleasure to take this opportunity to
thank the organizers and all the participants of the conference ``Interactions
between Geometry and Physics''. Their warm reception to a talk based on this
material persuaded me to dig out these notes from the depths of my hard disk
and finally submit them for publication.

A special word of thanks goes to all the people with whom I discussed
noncommutative matters during the last few years. The list includes (but is
not limited to) Igor Mencattini, Daniel Levcovitz, Fabio Ferrari Ruffino,
Pietro Tortella, Carlos Grossi, Sasha Anan'in, Claudio Bartocci, Valeriano
Lanza, Claudio L. S. Rava, Vladimir Rubtsov, Giovanni Landi and Andrea
Raimondo. I would also like to thank one of the anonymous referees for many
useful remarks.

A first draft of these notes was written when the author was supported by the
FAPESP post-doctoral grant 2011/09782-6.

\section{Differential calculus on associative algebras}
\label{s:ncdf}

In this section we review the construction of the universal calculus of
differential forms on associative algebras. Most of this material is taken
from Ginzburg's lectures \cite{ginzlect} and may be found in many other
standard references such as, for instance, \cite{landi97,dub01,gbvf01}.

\subsection{K\"ahler differentials}
\label{K-diff}

In ordinary (commutative) geometry a fundamental role is played by
\emph{differential forms}. As the name implies, the concept originated in the
theory of smooth spaces (differentiable manifolds), but was later exported in
much more general settings by a purely algebraic construction, known as
\emph{K\"ahler differentials}. It turns out that this more abstract
reformulation can easily be adapted to the associative context.

We start by recalling some definitions. Let \(A\) be a (not necessarily
commutative) ring. An \(A\)\textbf{-bimodule} is an abelian group \(M\)
equipped with \emph{two} actions of \(A\), one from the left and one from the
right, which are compatible in the following sense:
\begin{equation}
  \label{eq:05}
  a.(m.b) = (a.m).b \qquad\text{ for every } a,b\in A,\, m\in M.
\end{equation}
An \(A\)-bimodule \(M\) is called \textbf{symmetric} if the two actions
coincide, that is \(a.m = m.a\) for every \(a\in A\) and \(m\in M\). Clearly,
if \(A\) is commutative then every left (or right) \(A\)-module can be
extended to a symmetric \(A\)-bimodule by simply defining the opposite action
to be equal to the one given (but notice that, even in the commutative case,
not every \(A\)-bimodule is of this kind). It follows that the categories of
left \(A\)-modules, of right \(A\)-modules and of symmetric \(A\)-bimodules
are all isomorphic when the ring \(A\) is commutative, hence we can simply
speak of ``\(A\)-modules'' and let the elements of \(A\) act from whatever
side we wish.

Now let \(\bb{K}\) be a field of characteristic zero and \(A\) a commutative
\(\bb{K}\)-algebra, thought of as the algebra of ``regular'' functions \(X\to
\bb{K}\) for some space \(X\). Given an \(A\)-module \(M\), a
\(\bb{K}\)-linear map \(\app{\delta}{A}{M}\) is called a \textbf{derivation}
if
\begin{equation}
  \label{eq:06}
  \delta(ab) = \delta(a)b + a\delta(b) \qquad\text{ for every } a,b\in A.
\end{equation}
Let us denote by \(\Der(A,M)\) the set of derivations \(A\to M\), and consider
the functor
\begin{equation}
  \label{eq:07}
  \app{\Der(A,\mph)}{\lmod{A}}{\mathbf{Set}}
\end{equation}
which sends each \(A\)-module \(M\) to the set \(\Der(A,M)\) and each
\(A\)-linear map \(\app{f}{M}{N}\) to its pushforward \(f^{*}\) (defined by
\(f^{*}(\delta) = f\circ \delta\)). We claim that this functor has a
\emph{universal element} (see \cite[p. 57]{MacLane1998}), namely there exists
a pair \(\p{\Omega^{1}(A),d}\), where \(\Omega^{1}(A)\) is an \(A\)-module and
\(d\in \Der(A,\Omega^{1}(A))\), such that for every other pair
\(\p{M,\theta}\) of this kind there exists a unique \(A\)-linear map
\(\app{f_{\theta}}{\Omega^{1}(A)}{M}\) that makes the following diagram
commute in \(\lmod{A}\):
\begin{equation}
  \label{eq:08}
  \xymatrix{
    A\ar[r]^-{d} \ar[dr]_{\theta} & \Omega^{1}(A) \ar[d]^{f_{\theta}}\\
    & M}
\end{equation}
To see this one can either build the required \(A\)-module ``by hand'', taking
as generators the formal symbols \(\{\de a\}_{a\in A}\) and imposing suitable
relations between them (see e.g. \cite[Chap. 16]{Eisenbud1995}), or proceed in
a more explicit way by taking a suitable quotient of the kernel of the
multiplication map of \(A\), seen as a \(\bb{K}\)-linear map\footnote{Here and
  in what follows we adopt the following convention: whenever we use a tensor
  product sign without a subscript, we mean a tensor product over \(\bb{K}\).}
\(A\otimes A\to A\) (see e.g. \cite[Chap. 9]{Matsumur1989}).

The elements of \(\Omega^{1}(A)\) are called \textbf{K\"ahler differentials}
and act as substitutes of differential forms in a purely algebraic context.
When \(A\) is the coordinate algebra of a \emph{smooth} algebraic variety, the
K\"ahler differentials are precisely the usual differential forms with regular
(that is, polynomial) coefficients; for \emph{singular} varieties the two
concepts no longer agree, and K\"ahler differentials behave in a better way%
\footnote{We remark that more care is needed when interpreting K\"ahler
  differentials in non-algebraic contexts, see
  \url{http://mathoverflow.net/q/6074}.}

Suppose now that our \(\bb{K}\)-algebra \(A\) is no longer commutative. Since
the definition \eqref{eq:06} of derivation makes perfect sense also when \(M\)
is a non-symmetric \(A\)-bimodule, we can again define a functor
\begin{equation}
  \label{eq:09}
  \app{\Der(A,\mph)}{\bimod{A}}{\mathbf{Set}}
\end{equation}
in exactly the same manner as above, the only difference being that now the
domain is the whole category of \(A\)-bimodules. We can then ask if the same
universal problem previously used to define K\"ahler differentials has a
solution in the new context, and the answer is affirmative.
\begin{teo}
  \label{th:omega1nc}
  The functor \eqref{eq:09} has a universal element: there exists a pair
  \(\p{\Omega^{1}_{\nc}(A),d}\), where \(\Omega^{1}_{\nc}(A)\) is an
  \(A\)-bimodule and \(d\in \Der(A,\Omega^{1}_{\nc}(A))\), such that for every
  other pair \(\p{M,\theta}\) of this kind there exists a unique
  \(A\)-bimodule morphism \(\app{f_{\theta}}{\Omega^{1}_{\nc}(A)}{M}\) that
  makes the following diagram commute in \(\bimod{A}\):
  \begin{equation}
    \label{eq:10}
    \xymatrix{
      A\ar[r]^-{d} \ar[dr]_{\theta} & \Omega^{1}_{\nc}(A) \ar[d]^{f_{\theta}}\\
      & M}
  \end{equation}
\end{teo}
The proof of this result is not difficult but would entail a long detour
through topics like Hochschild cohomology which will have no further use in
these notes. For this reason we omit it, referring the interested reader to
\cite[Section 10]{ginzlect}.

Concretely, if we denote by \(\app{\mu}{A\otimes A} {A}\) the multiplication
map of \(A\) we can take \(\Omega^{1}_{\nc}(A)\) to be the kernel of \(\mu\)
(seen as a sub-\(A\)-bimodule of \(A\otimes A\), the bimodule structure being
given by \(a.(a'\otimes a'') = aa'\otimes a''\), \((a'\otimes a'').b =
a'\otimes a''b\)) with the map \(d\) defined by
\begin{equation}
  \label{eq:14}
  da\deq 1\otimes a - a\otimes 1 \qquad\text{ for every } a\in A.
\end{equation}
By construction, every \(\alpha\in \Omega^{1}_{\nc}(A)\) may be written as a
finite sum of the form
\begin{equation}
  \label{eq:14a}
  \alpha = \sum_{i} a'_{i}\otimes a''_{i}
\end{equation}
for some \(a'_{i},a''_{i}\in A\) such that \(\sum_{i} a'_{i}a''_{i} = 0\).
Using the map \eqref{eq:14} we can write equivalently
\[ \alpha = \sum_{i} a'_{i}.da''_{i} = \sum_{i} a'_{i}.(1\otimes a''_{i} -
a''_{i}\otimes 1) = \sum_{i} a'_{i}\otimes a''_{i} - \sum_{i}
a'_{i}a''_{i}\otimes 1 = \sum_{i} a'_{i}\otimes a''_{i}. \]
Given a pair \(\p{M,\theta}\) as in the statement of theorem
\ref{th:omega1nc}, the map \(f_{\theta}\) is then defined by sending the
element \eqref{eq:14a} to \(\sum_{i} a'_{i}.\theta(a''_{i})\in M\).

It is useful to reformulate the universal property expressed by theorem
\ref{th:omega1nc} as the existence of a natural isomorphism
\begin{equation}
  \label{eq:14b}
  \Der(A,\mph)\isom \bimod{A}(\Omega^{1}_{\nc}(A),\mph).
\end{equation}
Again, this is exactly analogous to what happens in the commutative case,
with the category \(\bimod{A}\) replacing the category \(\lmod{A}\).

Particularly important is the case when \(M\) is the algebra \(A\) itself
seen as an \(A\)-bimodule in the obvious way, that is by defining
\[ a.x\deq ax \quad\text{ and }\quad x.b\deq xb \qquad\text{ for all }
a,b,x\in A. \]
When \(A\) is (commutative and) the algebra of regular functions on a smooth
affine manifold \(X\), the derivations \(A\to A\) are in 1-1 correspondence
with \emph{algebraic vector fields} globally defined on \(X\). As we shall see
in section \ref{s:rep}, the same interpretation makes sense also in the
associative context; we then take \(\Der(A)\deq \Der(A,A)\) to be the (linear)
space of vector fields on the ``associative variety'' determined by \(A\). The
natural isomorphism \eqref{eq:14b} then implies that for every \(\theta\in
\Der(A)\) there exists a unique \(A\)-bimodule map \(\app{i_{\theta}}
{\Omega^{1}_{\nc}(A)}{A}\) such that \(\theta = i_{\theta}\circ d\), that is
\begin{equation}
  \label{eq:31}
  i_{\theta}(da) = \theta(a) \quad\text{ for every } a\in A.
\end{equation}
Clearly this property specifies completely the action of \(i_{\theta}\) on
every element of \(\Omega^{1}_{\nc}(A)\).

Let us remark that the \(\bb{K}\)-linear space \(\Der(A)\) has a natural
structure of Lie algebra when equipped with the usual commutator bracket:
\[ [\theta_{1},\theta_{2}]\deq \theta_{1}\circ\theta_{2} - \theta_{2}\circ
\theta_{1}. \]
However, it cannot be equipped with the structure of a (left or right)
\(A\)-module as soon as \(A\) fails to be commutative. The best one can do is
to define an action of \(Z(A)\), the \emph{center} of the algebra \(A\), on
\(\Der(A)\) as follows: given \(k\in Z(A)\) and \(\theta\in \Der(A)\), we
take as \(k.\theta\) the map
\[ a\mapsto k\theta(a) \quad\text{ for every } a\in A. \]
This map is a derivation because
\[ (k.\theta)(ab) = k(\theta(a) b + a\theta(b)) = k\theta(a) b + ka\theta(b) \]
which coincides with
\[ (k.\theta)(a) b + a (k.\theta)(b) = k\theta(a)b + ak\theta(b) \]
precisely because \(ka = ak\) for every \(a\in A\). This makes \(\Der(A)\) a
symmetric \(Z(A)\)-bimodule.

\subsection{The complex $\Omega^{\bullet}(A)$}
\label{univ-dga}

From now on we are going to denote the bimodule of K\"ahler differentials for
a generic associative algebra \(A\) simply as \(\Omega^{1}(A)\). In order to
obtain a notion of \(n\)-form for every \(n>1\) we would like to build a
cochain complex \(\Omega^{\bullet}(A)\) whose lower degree part reduces to the
universal derivation \(\app{d}{A}{\Omega^{1}(A)}\) provided by theorem
\ref{th:omega1nc}. It turns out that the most convenient way to achieve this
goal involves another kind of universal construction.

Recall that a \(\bb{K}\)-algebra \(A\) is said to be \textbf{graded (over}
\(\bb{N}\)\textbf{)} if it comes equipped with a direct sum decomposition
\begin{equation}
  \label{eq:11}
  A = \bigoplus_{i\in \bb{N}} A_{i}
\end{equation}
such that \(A_{i}A_{j}\subseteq A_{i+j}\). It follows that \(A_{0}\) is a
subalgebra of \(A\) and each \(A_{i}\) is an \(A_{0}\)-bimodule (not
necessarily symmetric). A grading over \(\bb{N}\) automatically induces also a
\(\bb{Z}/2\) grading, namely a decomposition of \(A\) into an ``even'' and an
``odd'' part,
\begin{equation}
  \label{eq:15}
  A_{+}\deq \bigoplus_{i\text{ even}} A_{i} \quad\text{ and }\quad
  A_{-}\deq \bigoplus_{i\text{ odd}} A_{i},
\end{equation}
such that \(A_{\pm}A_{\pm}\subseteq A_{+}\), \(A_{\pm}A_{\mp}\subseteq
A_{-}\).

A \(\bb{K}\)-linear map \(\app{f}{A}{A}\) is said to be \textbf{of degree}
\(\ell\) if \(f(A_{i})\subseteq A_{i+\ell}\). A map \(\app{\delta}{A}{A}\) of
degree \(\ell\) which satisfies the \emph{graded Leibniz rule}
\begin{equation}
  \label{eq:13}
  \delta(ab) = \delta(a) b + (-1)^{k\ell} a \delta(b) \qquad\text{ for every }
  a\in A_{k},\, b\in A
\end{equation}
is called a \textbf{graded derivation of degree} \(\ell\). Ordinary
derivations of \(A\) are exactly the graded derivations of degree zero. We
shall speak of \emph{odd derivations} for graded derivations of odd degree,
and similarly for \emph{even derivations}.

The following result is easily proved using the graded Leibniz rule and
induction.
\begin{lem}
  \label{lem:der-gener}
  Let \(A\) be a graded algebra. If two derivations \(A\to A\) of a fixed
  degree coincide on a set of generators for \(A\) then they are equal.
\end{lem}
A \textbf{morphism of graded algebras} from \(A\) to \(B\) is a morphism of
\(\bb{K}\)-algebras \(\app{f}{A}{B}\) which has degree zero
(\(f(A_{i})\subseteq B_{i}\) for every \(i\in \bb{N}\)).

A \textbf{differential graded algebra}, or \textbf{dg-algebra} for short, is a
pair \(\p{D,d}\) consisting of a graded algebra \(D\) and a derivation
\(\app{d}{D}{D}\) of degree 1 such that \(d\circ d = 0\). The map \(d\) is
called the \emph{differential} of the dg-algebra; we shall denote by
\(\app{d_{n}}{D_{n}}{D_{n+1}}\) the restriction of \(d\) to the homogeneous
component of degree \(n\) in \(D\).

A \textbf{morphism of dg-algebras} from \(\p{D,d}\) to \(\p{E,d'}\) is a
morphism of graded algebras \(\app{f}{D}{E}\) that intertwines the two
differentials, which means that the diagram
\begin{equation}
  \label{eq:16}
  \xymatrix{
    D_{n} \ar[r]^{d_{n}} \ar[d]_{\rist{f}{D_{n}}} & D_{n+1} \ar[d]^{\rist{f}{D_{n+1}}}\\
    E_{n} \ar[r]_{d'_{n}} & E_{n+1}}
\end{equation}
commutes for every \(n\in \bb{N}\). The category obtained by taking as object
the dg-algebras over \(\bb{K}\) and as arrows the dg-algebra morphisms will be
denoted by \(\dgalg{\bb{K}}\).

As we already noted, the degree zero part of a graded \(\bb{K}\)-algebra is
itself a \(\bb{K}\)-algebra. It follows that there exists a restriction
functor
\begin{equation}
  \label{eq:17}
  \app{(\mph)_{0}}{\dgalg{\bb{K}}}{\alg{\bb{K}}}
\end{equation}
which sends a generic dg-algebra \(\p{D,d}\) to its degree zero part \(D_{0}\)
and a morphism of dg-algebras \(\app{f}{(D,d)}{(E,d')}\) to the restriction
\(\rist{f}{D_{0}}\) (which can be seen as a map with codomain \(E_{0}\)
because \(f\) preserves the grading). It turns out that this functor possess a
left adjoint, which enables us to canonically construct a dg-algebra extending
any given \(\bb{K}\)-algebra of ``degree zero'' elements.
\begin{teo}
  For every \(\bb{K}\)-algebra \(A\) there exists a universal morphism from
  \(A\) to \((\mph)_{0}\), that is a pair \(\p{\mathcal{D}(A),i}\) consisting
  of a dg-algebra \(\mathcal{D}(A)\) and a \(\bb{K}\)-algebra morphism
  \(\app{i}{A}{\mathcal{D}(A)_{0}}\) such that for every other pair
  \(\p{\Gamma,\psi}\) of this kind there exists a unique morphism of
  dg-algebras \(\app{u_{\psi}}{\mathcal{D}(A)}{\Gamma}\) that makes the
  following diagram commute in \(\alg{\bb{K}}\):
  \begin{equation}
    \label{eq:18}
    \xymatrix{
      A\ar[r]^-{i} \ar[dr]_{\psi} & \mathcal{D}(A)_{0}\ar[d]^{u_{\psi}}\\
      & \Gamma_{0}}
  \end{equation}
\end{teo}
This result was first proved in a seminal paper by Cuntz and Quillen
\cite{cq95}. The dg-algebra \(\mathcal{D}(A)\) is called the \textbf{universal
  differential envelope} of the \(\bb{K}\)-algebra \(A\) and admits a very
explicit description that we illustrate next.

Let us set \(\bar{A}\deq A/\bb{K}\), as a quotient of vector spaces over
\(\bb{K}\); for every \(a\in A\) we shall denote by \(\ol{a}\) its image along
the canonical projection \(A\to \bar{A}\). Now define\footnote{Recall our
  convention about tensor products: a \(\otimes\) sign means by default
  \(\otimes_{\bb{K}}\).}
\begin{equation}
  \label{eq:19}
  \mathcal{D}(A)_{n}\deq A\otimes
  \underbrace{\bar{A}\otimes\dots\otimes\bar{A}}_{n\text{ times}}
\end{equation}
and let \(\mathcal{D}(A)\deq \bigoplus_{n\in \bb{N}} \mathcal{D}(A)_{n}\).
Then define \(\app{d_{n}}{\mathcal{D}(A)_{n}}{\mathcal{D}(A)_{n+1}}\) as
follows:
\begin{equation}
  \label{eq:20}
  d_{n}(a_{0}\otimes \ol{a_{1}}\otimes \dots\otimes \ol{a_{n}})\deq
  1\otimes \ol{a_{0}}\otimes \dots\otimes \ol{a_{n}}.
\end{equation}
Finally we must give a product on \(\mathcal{D}(A)\) compatible with the
grading. We do this by defining a map \(\mathcal{D}(A)_{n}\times
\mathcal{D}(A)_{m-1-n}\to \mathcal{D}(A)_{m-1}\) via the following prescription:
\begin{multline}
  \label{eq:21}
  (a_{0}\otimes \ol{a_{1}}\otimes \dots\otimes \ol{a_{n}}) (a_{n+1}\otimes
  \ol{a_{n+2}}\otimes \dots\otimes \ol{a_{m}})\deq \\
  (-1)^{n} a_{0}a_{1}\otimes \ol{a_{2}} \otimes\dots\otimes \ol{a_{m}} +
  \sum_{i=1}^{n} (-1)^{n-i} a_{0}\otimes \ol{a_{1}}\otimes \dots \otimes
  \ol{a_{i}a_{i+1}}\otimes \dots\otimes \ol{a_{m}},
\end{multline}
where the \(a_{1}\) that figures in the first term on the right-hand side is
any representative for \(\ol{a_{1}}\) (it is easy to check that the result
does not depend on the particular representative chosen).

Notice in particular the non-trivial action of an element of
\(\mathcal{D}(A)_{0} = A\) on \(\mathcal{D}(A)_{n}\) from the right (the
action from the left is the obvious one):
\begin{multline*}
  \label{eq:42}
  (a_{0}\otimes \ol{a_{1}}\otimes\dots\otimes \ol{a_{n}}) a_{n+1} = 
  (-1)^{n} a_{0}a_{1}\otimes \ol{a_{2}} \otimes\dots\otimes \ol{a_{n+1}} + {}\\
  {} + (-1)^{n-1} a_{0}\otimes \ol{a_{1}a_{2}} \otimes\dots\otimes \ol{a_{n+1}} +
  \dots + a_{0}\otimes \ol{a_{1}}\otimes \dots \otimes \ol{a_{n}a_{n+1}}.
\end{multline*}
For instance when \(n=1\) we have
\begin{equation*}
  a(a_{0}\otimes \ol{a_{1}}) = aa_{0}\otimes \ol{a_{1}} \quad\text{ but }\quad
  (a_{0}\otimes \ol{a_{1}})b = -a_{0}a_{1}\otimes \ol{b} + a_{0}\otimes
  \ol{a_{1}b}.
\end{equation*}
In \cite{cq95} the authors prove (1) that the pair \(\p{\mathcal{D}(A),d}\) is
a dg-algebra, and (2) that the pair \(\p{\mathcal{D}(A),i}\), where \(i\) is
the identity map \(A\to \mathcal{D}(A)_{0} = A\), solves the universal problem
\eqref{eq:18}.

Let us clarify the relationship between \(\mathcal{D}(A)\) and the K\"ahler
differentials of \(A\). Consider the degree zero component of the map
\eqref{eq:20}; as \(\mathcal{D}(A)_{0} = A\), it is a derivation from \(A\) to
the \(A\)-bimodule \(\mathcal{D}(A)_{1}\). It can be shown (see \cite[Theorem
10.7.1]{ginzlect}) that the pair \(\p{\mathcal{D}(A)_{1}, d_{0}}\) is
\emph{also} a universal element for the functor \(\Der(A,\mph)\). The
universal property of the pair \(\p{\Omega^{1}(A),d}\) then implies that there
exists a unique isomorphism of \(A\)-bimodules
\[ \app{\psi}{\Omega^{1}(A)}{\mathcal{D}(A)_{1}} \]
such that \(\psi\circ d = d_{0}\); that is, \(\psi\) sends \(da = 1\otimes a -
a\otimes 1\in \Omega^{1}(A)\) to \(1\otimes \ol{a}\in A\otimes \bar{A} =
\mathcal{D}(A)_{1}\). Its inverse is given by
\[ a_{0}\otimes \ol{a_{1}}\mapsto a_{0}\otimes a_{1} - a_{0}a_{1}\otimes 1 \]
(it is easy to verify that the right-hand side does not depend on the
particular representative chosen for \(\ol{a_{1}}\)). We conclude that the
\(A\)-bimodule of degree 1 elements in \(\mathcal{D}(A)\) gives another
realization of the K\"ahler differentials for \(A\).

To understand the nature of the complex \(\mathcal{D}(A)\) in higher degrees
let us briefly review the notion of \emph{tensor algebra of a bimodule}. Given
a \(\bb{K}\)-algebra \(A\) and an \(A\)-bimodule \(M\), the \textbf{tensor
  algebra} of \(M\) is defined to be the \(\bb{K}\)-vector space
\begin{equation}
  \label{eq:24}
  \bm{T}_{A}(M)\deq \bigoplus_{i\in \bb{N}} \bm{T}_{A}^{i}(M)
\end{equation}
where \(\bm{T}_{A}^{0}(M)\deq A\), \(\bm{T}_{A}^{1}(M)\deq M\) and, for every
\(i>1\),
\[ \bm{T}_{A}^{i}(M)\deq \underbrace{M\otimes_{A} \dots \otimes_{A}
  M}_{i\text{ times}}. \]
Clearly \(\bm{T}_{A}(M)\) is a \(\bb{K}\)-algebra graded over \(\bb{N}\), by
taking \(\bm{T}_{A}^{i}(M)\) as the homogeneous component of degree \(i\). It
will be useful to interpret the tensor algebra construction as an adjoint
functor; to explain this, however, a little aside is necessary.

In commutative algebra, by an \emph{algebra over a commutative ring} \(R\) it
is usually meant a pair \(\p{A,\eta}\) where \(A\) is a (not necessarily
commutative) ring and \(\app{\eta}{R}{A}\) is a morphism of rings \emph{whose
  image is contained in the center of} \(A\) (see e.g. \cite[p.
121]{Lang2002}). This gives to \(A\) the structure of a left \(R\)-module
(or right \(R\)-module, or symmetric \(R\)-bimodule) by defining \(r.a\deq
\eta(r)a\) and/or \(a.r\deq a\eta(r)\); the two expressions always coincide
precisely because \(\eta(r)\) belongs to the center of \(A\).

When \(R\) is no longer assumed to be commutative it is natural (and in fact
necessary, if we want non-symmetric bimodules) to drop the constraint on the
image of \(\eta\); hence for us an \textbf{algebra over the ring} \(R\) will
be a pair \(\p{A,\eta}\) consisting of a ring \(A\) and a morphism of rings
\(\app{\eta}{R}{A}\). Again, this means that \(A\) is automatically equipped
with the structure of a (not necessarily symmetric) \(R\)-bimodule defined by
\begin{equation}
  \label{eq:25}
  r.a.r' = \eta(r)\, a\, \eta(r') \quad\text{ for every } a\in A,\, r,r'\in R.
\end{equation}
By an unfortunate mismatch in terminology, the category of algebras over \(R\)
(in this ``non-central'' sense) is exactly what category theorists call the
\emph{category of objects under} \(R\) \emph{in} \(\mathbf{Rng}\) (see e.g.
\cite[p. 45]{MacLane1998}); we shall denote it by \(R\downarrow \mathbf{Rng}\)%
\footnote{Another popular notation is \(R/\mathbf{Rng}\). Notice that we keep
  using the notation \(\alg{R}\) for the category of ``usual'' algebras over a
  commutative ring (or field) \(R\).}.

With this confusing point understood, let us return to the tensor algebra of
an \(A\)-bimodule \(M\). Since the map
\begin{equation}
  \label{eq:29}
  \app{\eta}{A}{\bm{T}_{A}(M)}
\end{equation}
defined by sending each \(a\in A\) to the corresponding degree zero tensor is
an (injective) morphism of \(\bb{K}\)-algebras, the tensor algebra
\(\bm{T}_{A}(M)\) is naturally an algebra over \(A\), that is an object of the
category \(A\downarrow \alg{\bb{K}}\). It follows that we can define a functor
\begin{equation}
  \label{eq:26}
  \app{\bm{T}_{A}}{\bimod{A}}{A\downarrow \alg{\bb{K}}}
\end{equation}
by sending each \(A\)-bimodule \(M\) to its tensor algebra and each morphism
of \(A\)-bimodules \(\app{f}{M}{N}\) to the map
\(\app{\bm{T}_{A}(f)}{\bm{T}_{A}(M)}{\bm{T}_{A}(N)}\) defined on decomposable
tensors of degree \(i\) as
\begin{equation}
  \label{eq:28}
  m_{1}\otimes_{A} \dots \otimes_{A} m_{i}\mapsto f(m_{1})\otimes_{A} \dots
  \otimes_{A} f(m_{i})
\end{equation}
(as is well known, tensors of this form generate the whole of
\(\bm{T}_{A}^{i}(M)\)). In the other direction we have the ``partially
forgetful'' functor \(\app{U}{A\downarrow \alg{\bb{K}}}{\bimod{A}}\) that,
given an algebra \(B\) over \(A\), forgets the product in \(B\) but keeps the
\(A\)-bimodule structure. It is easy to check that, for every \(A\)-bimodule
\(M\) and every \(A\)-algebra \(B\), there is a natural isomorphism
\begin{equation}
  \label{eq:27}
  A\downarrow \alg{\bb{K}}(\bm{T}_{A}(M),B) \isom \bimod{A}(M,U(B))
\end{equation}
and this means exactly that the functor \(\bm{T}_{A}\) is left adjoint to
\(U\). Hence the tensor algebra construction \eqref{eq:24} can be seen as the
universal way to ``enhance'' the structure of an \(A\)-bimodule to that of a
full algebra over \(A\). Again, this is exactly analogous to what happens
in the corresponding commutative situation, where the tensor algebra of a
module gives a functor \(\lmod{R}\to \alg{R}\) which is left adjoint to the
forgetful functor \(\alg{R}\to \lmod{R}\).
\begin{lem}
  \label{lem:ext-Alin}
  Let \(M\) be an \(A\)-bimodule and \(B\) be an algebra over \(A\). Every
  morphism of \(A\)-bimodules \(\app{f}{M}{B}\) can be extended in a unique
  manner to a (even or odd) derivation \(\app{\delta_{f}}{\bm{T}_{A}(M)}{B}\)
  such that:
  \begin{enumerate}
  \item the restriction of \(\delta_{f}\) to \(\bm{T}_{A}^{0}(M)\isom A\) is a
    specified map \(A\to B\), and
  \item the restriction of \(\delta_{f}\) to \(\bm{T}_{A}^{1}(M)\isom M\)
    coincides with \(f\).
  \end{enumerate}
\end{lem}
The proof boils down to a straightforward induction, where the inductive step
uses the (possibly graded) Leibniz rule to reduce by one the degree of the
tensor on which \(\delta_{f}\) operates.

Consider now the tensor algebra determined by the bimodule of K\"ahler
differentials of \(A\),
\begin{equation}
  \label{eq:30}
  \Omega^{\bullet}(A)\deq \bm{T}_{A}(\Omega^{1}(A)).
\end{equation}
A tensor of degree \(k\) in \(\Omega^{\bullet}(A)\) may be written as a linear
combination of terms of the form
\[ a_{1}.db_{1} \otimes_{A} a_{2}.db_{2} \otimes_{A} \dots \otimes_{A}
a_{k}.db_{k}.a_{k+1} \]
where each \(a_{i}\), \(2\leq i\leq k\) may be freely moved across the tensor
product sign. Clearly this looks quite different from a typical element of
\(\mathcal{D}(A)_{k}\). This notwithstanding, we have the following:
\begin{teo}
  \label{th:iso-Om-D}
  The graded algebras \(\Omega^{\bullet}(A)\) and \(\mathcal{D}(A)\) are
  isomorphic.
\end{teo}
\begin{proof}
  Let \(\app{\psi}{\Omega^{1}(A)}{\mathcal{D}(A)_{1}}\) be the \(A\)-bimodule
  isomorphism defined above. By composing \(\psi\) with the embedding
  \(\mathcal{D}(A)_{1}\hookrightarrow \mathcal{D}(A)\) we obtain an
  \(A\)-bimodule map \(\Omega^{1}(A)\to \mathcal{D}(A)\). Using the natural
  isomorphism \eqref{eq:27}, this corresponds to a morphism
  \(\app{\Psi}{\bm{T}_{A}(\Omega^{1}(A))}{\mathcal{D}(A)}\) of algebras over
  \(A\) which is defined on decomposable tensors by
  \[ \Psi(\alpha_{1}\otimes_{A} \dots \otimes_{A} \alpha_{n}) =
  \psi(\alpha_{1})\dots \psi(\alpha_{n}) \]
  To conclude it is sufficient to show that \(\Psi\) is invertible. Its
  inverse can be defined as follows: given a decomposable element
  \(a_{0}\otimes \ol{a_{1}}\otimes \dots \otimes \ol{a_{n}}\) in
  \(\mathcal{D}(A)_{n}\), we write it as the product
  \[ (a_{0}\otimes \ol{a_{1}})(1\otimes \ol{a_{2}}) \dots (1\otimes \ol{a_{n}}) \]
  and send it to
  \[ \psi^{-1}(a_{0}\otimes \ol{a_{1}})\otimes_{A} \psi^{-1}(1\otimes
  \ol{a_{2}})\otimes_{A} \dots \otimes_{A} \psi^{-1}(1\otimes \ol{a_{n}}) \]
  It is clear that the map so defined is an inverse for \(\Psi\).
\end{proof}
We can use the isomorphism \(\Psi\) to transfer the differential \eqref{eq:20}
naturally defined on \(\mathcal{D}(A)\) on the tensor algebra
\(\Omega^{\bullet}(A)\), as the map \(\de\deq \Psi^{-1}\circ d\circ \Psi\).
For instance, a generic element \(\alpha = \sum_{i} a_{i}.db_{i}\) in
\(\Omega^{1}(A)\) corresponds via the map \(\psi\) to \(\sum_{i} a_{i}\otimes
\ol{b_{i}}\) in \(\mathcal{D}(A)_{1}\). Its differential in
\(\mathcal{D}(A)_{2}\) is
\[ \sum_{i} 1\otimes \ol{a_{i}}\otimes \ol{b_{i}} = \sum_{i} (1\otimes
\ol{a_{i}})(1\otimes \ol{b_{i}}) \]
and \(\Psi^{-1}\) maps this element back to
\begin{equation}
  \label{eq:01}
  \de\alpha = \sum_{i} da_{i}\otimes_{A} db_{i}\in \Omega^{1}(A)\otimes_{A}\Omega^{1}(A)
\end{equation}
Notice that we could certainly use lemma \ref{lem:ext-Alin} to directly define
a map \(\app{\de}{\Omega^{\bullet}(A)}{\Omega^{\bullet}(A)}\) which restricts
to the universal derivation \eqref{eq:14} on tensors of degree zero and acts
as in equation \eqref{eq:01} on tensors of degree 1. However, it is then a
non-trivial endeavor to show that \(\de\circ \de = 0\). On the contrary, the
proof of this fact is immediate in the universal differential envelope. Thus,
thanks to theorem \ref{th:iso-Om-D} we can have the best of both worlds.

\subsection{The Karoubi-de Rham complex}
\label{assoc-df}

At this point it would seem natural to interpret the pair
\(\p{\Omega^{\bullet}(A),\de}\) as the complex of differential forms on the
``associative variety'' determined by the algebra \(A\). There are two
problems with this idea. The first one is that the cohomology of this complex
turns out to be (rather trivial and) entirely independent from \(A\), as the
following result shows.
\begin{teo}
  \label{th:coh-Om}
  The cohomology of the complex \(\p{\Omega^{\bullet}(A),\de}\) is given by
  \begin{equation}
    \label{eq:38}
    H^{k}(\Omega^{\bullet}(A)) =
    \begin{cases}
      \bb{K} &\text{ if } k=0\\
      0 &\text{ otherwise.}
    \end{cases}
  \end{equation}
\end{teo}
\begin{proof}
  It is convenient to work in \(\p{\mathcal{D}(A),d}\). From the expression
  \eqref{eq:20} it is clear that \(H^{0}(\Omega^{\bullet}(A)) = \ker d_{0} =
  \bb{K}1\). On the other hand, for every \(n>1\) the map \(d_{n}\) admits the
  following factorization:
  \begin{equation}
    \label{eq:39}
    \xymatrix{
      & \bb{K}1\otimes \bar{A}^{\otimes (n+1)} \ar@{^{(}->}[dr] &\\
      A\otimes \bar{A}^{\otimes n} \ar@{>>}[ur] \ar[rr]^{d_{n}} & & A\otimes
      \bar{A}^{\otimes (n+1)}}
  \end{equation}
  But \(\ker d_{n+1}\) is exactly \(\bb{K}1\otimes \bar{A}^{\otimes (n+1)}\),
  hence the triviality of \(H^{k}(\Omega^{\bullet}(A))\) for \(k>0\).
\end{proof}
The second drawback is that, even when the algebra \(A\) is commutative,
\(\Omega^{\bullet}(A)\) \emph{does not coincide} with the usual dg-algebra of
differential forms on the corresponding affine variety. Indeed, the algebra
\(\Omega^{\bullet}(A)\) is not graded-commutative: for instance \(\de a\, \de
b \neq -\de b\, \de a\) in general\footnote{Some applications of these
  ``non-standard'' differential forms on commutative algebras can be found in
  \cite{Mueller-Hoissen1997}.}. To recover this property we need to take a
quotient of \(\Omega^{\bullet}(A)\) in which such relations are imposed by
hand.

In order to do this let us recall that the \textbf{graded commutator} on a
graded algebra \(D\) is the map \(\app{\gc{\cdot}{\cdot}}{D\times D}{D}\)
defined on homogeneous elements \(a\in D_{i}\), \(b\in D_{j}\) by
\begin{equation}
  \label{eq:35}
  \gc{a}{b}\deq ab - (-1)^{ij} ba
\end{equation}
and then extended by linearity on the whole of \(D\). Notice that
\(\gc{a}{b}\) coincides with the ordinary commutator \([a,b]\) as soon as at
least one of \(a\) and \(b\) has even degree, whereas for two elements of odd
degree we have \(\gc{a}{b} = ab + ba\) instead. In the sequel the following
compatibility property between derivations and graded commutators, which is
immediate to check, will be rather useful.
\begin{lem}
  \label{lem:grc-der}
  Let \(D\) be a graded algebra and \(\app{\delta}{D}{D}\) be a graded
  derivation. Then \(\delta\) maps a graded commutator in a linear combination
  of graded commutators.
\end{lem}
Now let \(A\) be any associative algebra. The \textbf{Karoubi-de Rham complex}
of \(A\) \cite{kar86} is the graded vector space over \(\bb{K}\) given by the
quotient
\begin{equation}
  \label{eq:40}
  \DR^{\bullet}(A)\deq \frac{\Omega^{\bullet}(A)}
  {\gc{\Omega^{\bullet}(A)}{\Omega^{\bullet}(A)}}
\end{equation}
where \(\gc{\Omega^{\bullet}(A)}{\Omega^{\bullet}(A)}\) denotes the linear
subspace in \(\Omega^{\bullet}(A)\) generated by all the elements of the form
\(\gc{a}{b}\) for \(a,b\in \Omega^{\bullet}(A)\). This is indeed a graded
subspace, so that the quotient \eqref{eq:40} makes sense in the category of
graded vector spaces over \(\bb{K}\) (it does not make sense in the category
graded \(\bb{K}\)-algebras, since \(\gc{\Omega^{\bullet}(A)}
{\Omega^{\bullet}(A)}\) is not an ideal). We shall take the elements of
\(\DR^{\bullet}(A)\) as the associative-geometric counterpart of differential
forms.

In general it is not easy to explicitly describe these objects. In degree
zero, however, we have obviously
\begin{equation}
  \label{eq:50}
  \DR^{0}(A) = \frac{A}{\comm{A}{A}}
\end{equation}
where \(\comm{A}{A}\) is the linear subspace spanned by commutators in \(A\).
Classically a 0-form is just a function, so it is natural to regard
\(\DR^{0}(A)\) as the linear space of ``regular functions'' on the associative
variety determined by \(A\).

The homogeneous component \(\DR^{1}(A)\) of the quotient \eqref{eq:40} is
also easy to describe: as the only degree 1 relations in
\(\gc{\Omega^{\bullet}(A)} {\Omega^{\bullet}(A)}\) are of the form \(a\beta -
\beta a\) for some \(a\in A\) and \(\beta\in \Omega^{1}(A)\), we have that%
\footnote{\emph{Cognoscenti} will recognize \(\DR^{0}(A)\) and \(\DR^{1}(A)\)
  as the degree zero Hochschild homology of the \(A\)-bimodules \(A\) and
  \(\Omega^{1}(A)\), respectively.}
\begin{equation}
  \label{eq:51}
  \DR^{1}(A) = \frac{\Omega^{1}(A)}{\comm{A}{\Omega^{1}(A)}}.
\end{equation}
As soon as \(n\geq 2\) things get more complicated: for example \(\DR^{2}(A)\)
is defined by relations coming from both the subspaces
\[ \comm{A}{\Omega^{2}(A)} = \spa_{\bb{K}} \{a\omega - \omega a\}_{a\in
  A,\, \omega\in \Omega^{2}(A)} \]
and
\[ \gc{\Omega^{1}(A)}{\Omega^{1}(A)} = \spa_{\bb{K}}
\{\alpha\beta + \beta\alpha\}_{\alpha,\beta\in \Omega^{1}(A)}. \]
Clearly, without further information on the algebra \(A\) there is little hope
for an explicit description of these quotients.

It follows from lemma \ref{lem:grc-der} that the differential of the complex
\(\Omega^{\bullet}(A)\) maps a graded commutator to a linear combination of
graded commutators, and so descends to a map
\begin{equation}
  \label{eq:de-DR}
  \app{\de}{\DR^{\bullet}(A)}{\DR^{\bullet}(A)}.
\end{equation}
Moreover, this map still obeys the fundamental relation \(\de\circ \de = 0\).
This means that the pair \(\p{\DR^{\bullet}(A), \de}\) qualifies as a
\emph{differential graded vector space}, that is a graded vector space
equipped with a map increasing the degree by one and whose square vanishes. On
the other hand, it does \emph{not} qualify as a dg-algebra, since elements of
\(\DR^{\bullet}(A)\) cannot be meaningfully multiplied. (For the same reason,
the map \eqref{eq:de-DR} is not itself a derivation.)

However, the apparent lack of an ``exterior product'' operation between
associative differential forms is not as serious a problem as it may seem. The
reason is that many constructions involving such products can be performed at
the level of the dg-algebra \(\Omega^{\bullet}(A)\) and then pushed down to
its quotient \(\DR^{\bullet}(A)\). Let us show how this works in practice by
setting up a ``differential calculus'' for associative differential forms,
following \cite[Section 11]{ginzlect}.

As anticipated in \S\ref{K-diff}, the role of \emph{vector fields} will be
played by derivations \(A\to A\). For every vector field \(\theta\in \Der(A)\)
we have the \(A\)-bimodule map \(\app{i_{\theta}}{\Omega^{1}(A)}{A}\) defined
by the equality \eqref{eq:31}. By composing with the embedding
\(A\hookrightarrow \Omega^{\bullet}(A)\) we can see \(i_{\theta}\) as a
morphism of \(A\)-bimodules from \(\Omega^{1}(A)\) to \(\Omega^{\bullet}(A)\).
Then we can use lemma \ref{lem:ext-Alin} to extend this map to a derivation of
degree \(-1\) on \(\Omega^{\bullet}(A)\) (that is, an odd derivation mapping
each \(\Omega^{n}(A)\) to \(\Omega^{n-1}(A)\)) which vanishes on tensors of
degree zero. The resulting map
\begin{equation}
  \label{eq:32}
  \app{i_{\theta}}{\Omega^{\bullet}(A)}{\Omega^{\bullet}(A)}
\end{equation}
will be called the \textbf{interior product} on \(\Omega^{\bullet}(A)\). For
instance its action on a generic degree two elements is
\begin{equation}
  \label{eq:43}
  i_{\theta}(a_{1}\de b_{1} a_{2}\de b_{2} a_{3}) = a_{1}\theta(b_{1}) a_{2}
  \de b_{2} a_{3} - a_{1} \de b_{1} a_{2} \theta(b_{2}) a_{3}
\end{equation}
and one can readily verify that \(i_{\theta}\circ i_{\theta} = 0\). More
generally, we have
\begin{equation}
  \label{eq:33}
  i_{\theta}(\de a_{1}\dots \de a_{n}) = \sum_{i=1}^{n} (-1)^{i-1} \de a_{1}
  \dots \theta(a_{i}) \dots \de a_{n}.
\end{equation}
It follows from lemma \ref{lem:grc-der} that each map \(i_{\theta}\) descends
to a map on the Karoubi-de Rham complex that we still denote in the same way,
\[ \app{i_{\theta}}{\DR^{\bullet}(A)}{\DR^{\bullet}(A)}. \]
We notice that the following equality holds for every pair of derivations
\(\theta,\eta\in \Der(A)\):
\begin{equation}
  \label{eq:33b}
  i_{\theta}\circ i_{\eta} = -i_{\eta}\circ i_{\theta}.
\end{equation}
In particular the natural ``pairing'' map \(\Omega^{1}(A)\times \Der(A)\to A\)
defined by \(\p{\alpha,\theta}\mapsto i_{\theta}(\alpha)\) descends to a
\(\bb{K}\)-linear map
\begin{equation}
  \label{eq:pair-DR}
  \app{\pair{\mph}{\mph}}{\DR^{1}(A)\times \Der(A)}{\DR^{0}(A)}
\end{equation}
representing the operation of contraction between a 1-form and a vector field,
resulting in a regular function.

Now that we have both an exterior differential and an interior product on
\(\Omega^{\bullet}(A)\) it is straightforward to define a companion ``Lie
derivative'' operator using Cartan's formula:
\begin{equation}
  \label{eq:34}
  \mathcal{L}_{\theta}\deq \de\circ i_{\theta} + i_{\theta}\circ \de
\end{equation}
By definition, \(\mathcal{L}_{\theta}\) is a degree \(0\) (hence even)
derivation. Explicitly, it acts as follows:
\begin{equation}
  \label{eq:36}
  \mathcal{L}_{\theta}(a_{0}\de a_{1}\dots \de a_{n}) = \theta(a_{0}) \de
  a_{1}\dots \de a_{n} + \sum_{i=1}^{n} a_{0} \de a_{1} \dots \de
  \theta(a_{i}) \dots \de a_{n}
\end{equation}
Moreover, using lemma \ref{lem:der-gener} one can verify by a direct
calculation on 1-forms (which generate \(\Omega^{\bullet}(A)\) as an algebra)
that the following familiar identities hold:
\begin{gather}
  \label{eq:37b}
  \mathcal{L}_{\theta}\circ i_{\eta} - i_{\eta}\circ \mathcal{L}_{\theta} =
  i_{\comm{\theta}{\eta}}\\
  \label{eq:37a}
  \mathcal{L}_{\theta}\circ \mathcal{L}_{\eta} - \mathcal{L}_{\eta}\circ
  \mathcal{L}_{\theta} = \mathcal{L}_{\comm{\theta}{\eta}}
\end{gather}
By lemma \ref{lem:grc-der} each map \(\mathcal{L}_{\theta}\) descends to the
complex \(\DR^{\bullet}(A)\), where all the previous identities continue to
hold. In particular, the identity \eqref{eq:37b} applied to a 1-form
\(\alpha\in \DR^{1}(A)\) can be interpreted as an analogue of the classical
fact that ``Lie derivatives distribute inside contractions'':
\[ \mathcal{L}_{\theta}(i_{\eta}(\alpha)) =
i_{\eta}(\mathcal{L}_{\theta}(\alpha)) + i_{\comm{\theta}{\eta}}(\alpha) \]
where the commutator \([\theta,\eta]\) is interpreted as the action of
\(\mathcal{L}_{\theta}\) on the derivation \(\eta\).

Contrary to what happens for \(\Omega^{\bullet}(A)\), computing the cohomology
of the complex \(\DR^{\bullet}(A)\) for a given associative algebra \(A\) is
usually a nontrivial problem. The next result is sometimes useful in this
connection. It states that, when the algebra \(A\) itself is graded \emph{in
  positive degrees only}, each cohomology group of \(\DR^{\bullet}(A)\)
depends only on the subalgebra of degree zero elements in \(A\).
\begin{teo}
  \label{th:poinc-lem}
  Suppose the algebra \(A\) is graded over \(\bb{N}\). For every \(k\geq 0\)
  there is an isomorphism
  \begin{equation}
    \label{eq:53}
    H^{k}(\DR^{\bullet}(A_{0}))\isom H^{k}(\DR^{\bullet}(A)).
  \end{equation}
\end{teo}
The proof closely mimics the usual argument leading to the Poincaré lemma for
ordinary de Rham cohomology; the reader may find the details in
\cite{ginzlect} (Theorem 11.4.7).

\subsection{Associative affine space}
\label{nc-aff-sp}

As the first (and simplest) example of an associative variety we consider the
\emph{associative affine} \(n\)\emph{-dimensional space}, that is the
associative space which corresponds to the free associative algebra on \(n\)
generators:
\begin{equation}
  \label{eq:67}
  A = \falg{\bb{K}}{x_{1}, \dots, x_{n}}.
\end{equation}
To deal with this case it is useful to adopt the following ``coordinate-free''
approach. Let \(V\) be a \(n\)-dimensional vector space with basis \(\p{e_{1},
  \dots, e_{n}}\) and let \(\p{x_{1}, \dots, x_{n}}\) denote a basis of the
dual space \(V^{*}\). Then the tensor algebra of \(V^{*}\) (over \(\bb{K}\))
\[ \bm{T}(V^{*}) = \bb{K}\oplus V^{*}\oplus (V^{*}\otimes V^{*})\oplus
\dots \]
is isomorphic to \(A\), as the reader can easily check; the tensor product is
simply the concatenation of words in the letters \(x_{1}, \dots, x_{n}\). This
point of view is useful because the resulting formalism is automatically
invariant under every \emph{affine} automorphism of the algebra \(A\)%
\footnote{It goes without saying that free associative algebras have many more
  automorphisms other than affine ones; their description is a classic (and
  difficult) problem in associative algebra.}.

Now let \(M\) be an \(A\)-bimodule. Every \(\bb{K}\)-linear map \(V^{*}\to M\)
can be extended to a derivation \(A\to M\) using the Leibniz rule, and every
element of \(\Der(A,M)\) arises in this way (because the generators of \(A\)
belong to \(V^{*}\)). On the other hand, the duality theorem for
finite-dimensional vector spaces implies that the space of linear maps
\(V^{*}\to M\) is canonically isomorphic to \(M\otimes V\). We conclude that
\begin{equation}
  \label{eq:67a}
  \Der(A,M)\isom M\otimes V
\end{equation}
and \(\Omega^{1}(A)\) is just the free \(A\)-bimodule generated by \(V^{*}\),
\begin{equation}
  \label{eq:67b}
  \Omega^{1}(A)\isom A\otimes V^{*}\otimes A.
\end{equation}
The pairing map \(\Omega^{1}(A)\times \Der(A)\to A\) defined by
\(\p{\alpha,\theta}\mapsto i_{\theta}(\alpha)\) is then given by
\begin{equation}
  \label{eq:pair-falg}
  \pair{a\otimes \varphi\otimes b}{c\otimes v} = \pair{\varphi}{v}_{V}\, acb
\end{equation}
where \(\pair{\cdot}{\cdot}_{V}\) denotes the pairing between \(V\) and
\(V^{*}\).

Notice that \(\Omega^{1}(A)\) is indeed isomorphic to \(A\otimes \bar{A}\), as
per general results, since
\[ V^{*}\otimes A = V^{*}\otimes (\bb{K}\oplus V^{*}\oplus V^{*\otimes
  2}\oplus \dots) \isom \bigoplus_{i>0} V^{*\otimes i} \isom
\bm{T}(V^{*})/\bb{K} = \bar{A}. \]
It follows that, for every \(p\geq 1\),
\[ \Omega^{p}(A)\isom A\otimes \bar{A}^{\otimes p} \isom A\otimes
(V^{*}\otimes A)^{\otimes p}. \]
Now let us study the Karoubi-de Rham complex of \(A\) starting from its
component of degree zero,
\[ \DR^{0}(A) = \frac{A}{\comm{A}{A}}. \]
It is not difficult to prove that two words in \(A\) differ by a commutator if
and only if their letters are related by a cyclic permutation. It follows that
\(\DR^{0}(A)\) can be identified with the \(\bb{K}\)-linear space generated by
the \textbf{necklace words} in the letters \(x_{1}, \dots, x_{n}\), that is
ordinary words considered modulo cyclic permutations of their letters. These
are well known combinatorial objects (see e.g. \cite[Chapter 15]{cam94}).

In degree 1, quotienting the free bimodule \eqref{eq:67b} by the linear
subspace \(\comm{A}{\Omega^{1}(A)}\) implies that we can always move an
element of \(A\) acting from the right to the left, as
\[ a\otimes \varphi\otimes b = ba\otimes \varphi\otimes 1 + \comm{a\otimes
  \varphi\otimes 1}{b} \]
and the second term is killed by the projection onto \(\DR^{1}(A)\). It
follows that
\[ \DR^{1}(A)\isom A\otimes V^{*} \]
as a \(\bb{K}\)-linear space. In particular for every \(\varphi\in V^{*}\)
there is the 1-form \(\de\varphi = 1\otimes \varphi\in \DR^{1}(A)\). The
pairing \eqref{eq:pair-falg} becomes
\[ \pair{a\otimes \varphi}{c\otimes v} = \pair{\varphi}{v}_{V}\, ac \mod [A,A] \]
Unfortunately, even in this very special case there is no easy description of
a generic associative \(p\)-form for \(p\geq 2\). On the other hand, the
cohomology of the complex \(\DR^{\bullet}(A)\) is readily computed, as first
shown by Kontsevich in \cite{kont93}.
\begin{teo}
  \label{th:coh-sp-aff}
  Let \(A\) be a free algebra. Then
  \begin{equation}
    \label{eq:73}
    H^{k}(\DR^{\bullet}(A)) =
    \begin{cases}
      \bb{K} &\text{ for } k=0\\
      0 &\text{ for } k\geq 1
    \end{cases}
  \end{equation}
\end{teo}
This is an immediate consequence of theorem \ref{th:poinc-lem}. In fact \(A =
\bm{T}(V^{*})\) is exactly a positively graded algebra whose degree zero part
is \(\bb{K}\), hence there is an isomorphism 
\[ H^{k}(\DR^{\bullet}(A))\isom H^{k}(\DR^{\bullet}(\bb{K})) \]
and the right-hand side is trivial for \(k>0\) and equal to \(\bb{K}\) for
\(k=0\).

We also have a notion of ``partial derivative'' of a regular function along a
direction in the dual vector space \(V\). Indeed, for every \(v\in V\) we can
define a map \(\app{\partial_{v}}{\DR^{0}(A)}{A}\) by
\begin{equation}
  \label{eq:75}
  \partial_{v}f\deq \pair{\de f}{1\otimes v}.
\end{equation}
In particular when \(v = e_{j}\) is an element of the basis \(\p{e_{1}, \dots,
  e_{n}}\) the resulting map is called the \textbf{necklace derivative} with
respect to the generator \(x_{j}\). It is natural to denote this map by
\(\frac{\partial}{\partial x_{j}}\), as
\[ \frac{\partial}{\partial x_{j}} x_{i} = \pair{1\otimes x_{i}}{1\otimes e_{j}} =
\delta_{ij} \]
where \(\delta_{ij}\) is the Kronecker symbol (\(\delta_{ij}=1\) when \(i=j\),
\(0\) otherwise).

More generally, given a necklace word \(f\in \DR^{0}(A)\) represented by the
(ordinary) word \(x_{i_{1}}\dots x_{i_{\ell}}\) for a suitable set of indices
\(i_{1}, \dots, i_{\ell}\in \{1,\dots,n\}\), we have that
\[ \de f = \sum_{k=1}^{\ell} x_{i_{1}}\dots x_{i_{k-1}} \de x_{i_{k}}
x_{i_{k+1}}\dots x_{i_{\ell}} = \sum_{k=1}^{\ell} x_{i_{k+1}}\dots
x_{i_{\ell}} x_{i_{1}}\dots x_{i_{k-1}} \de x_{i_{k}} \]
where the second equality holds in \(\DR^{1}(A)\). It follows that
\begin{equation}
  \label{eq:77}
  \frac{\partial}{\partial x_{j}} f = \pair{\sum_{k=1}^{\ell}
    x_{i_{k+1}}\dots x_{i_{\ell}} x_{i_{1}}\dots x_{i_{k-1}} \otimes
    x_{i_{k}}}{1\otimes e_{j}} = \sum_{k=1}^{\ell} \delta_{i_{k}j}
  x_{i_{k+1}}\dots x_{i_{\ell}} x_{1}\dots x_{i_{k-1}}.
\end{equation}
It is easy to check that this result does not depend on the particular
representative chosen for \(f\).

Finally let us derive an analogue of the usual formula for the differential of
a function in terms of partial derivatives. Given \(f\in \DR^{0}(A)\) we have
\(\de f = 1\otimes \bar{f}\) and since \(\bar{f}\in \bar{A}\isom V^{*}\otimes
A\) we can write
\[ \bar{f} = \sum_{i=1}^{n} x_{i}\otimes a_{i} \qquad\text{ for some } a_{1},
\dots, a_{n}\in A. \]
Then \(\de f  = \sum_{i=1}^{n} 1\otimes x_{i}\otimes a_{i}\in \Omega^{1}(A)\),
which projects down to
\begin{equation}
  \label{eq:70}
  \de f = \sum_{i=1}^{n} a_{i}\otimes x_{i}
\end{equation}
in \(\DR^{1}(A)\). Substituting into \eqref{eq:75} with \(v=e_{i}\) we see
that the coefficients \(a_{i}\) are exactly the necklace derivatives
\(\frac{\partial f}{\partial x_{i}}\), so that
\begin{equation}
  \label{eq:71}
  \de f = \sum_{i=1}^{n} \frac{\partial f}{\partial x_{i}}\otimes x_{i}
  \qquad \text{ for every } f\in \DR^{0}(A).
\end{equation}

\subsection{The Quillen complex}
\label{quillen}

Let us consider the map \(\Omega^{1}(A)\to [A,A]\) given by
\[ a_{0}\de a_{1}\mapsto [a_{0},a_{1}]. \]
This is well-defined because if \(a_{1}' = a_{1} + \lambda\) for some
\(\lambda\in \bb{K}\) then \([a_{0},a_{1}'] = [a_{0},a_{1}]\) (as \(\bb{K}\)
is contained in the center of \(A\)). Moreover, given \(a\in A\) and \(\beta =
b_{0}\de b_{1}\in \Omega^{1}(A)\) we have that the commutator
\[ a\beta - \beta a = ab_{0}\de b_{1} - (b_{0} \de b_{1})a = ab_{0}\de b_{1} +
b_{0}b_{1}\de a - b_{0}\de (b_{1}a) \]
is sent to
\[ [ab_{0},b_{1}] + [b_{0}b_{1},a] - [b_{0},b_{1}a] = ab_{0}b_{1} -
b_{1}ab_{0} + b_{0}b_{1}a - ab_{0}b_{1} - b_{0}b_{1}a + b_{1}ab_{0} = 0. \]
We conclude that there is a well-defined map
\begin{equation}
  \label{eq:65}
  \app{b}{\DR^{1}(A)}{[A,A]}
\end{equation}
given by \(u\de v\mapsto [u,v]\). It is easy to check that
\begin{equation}
  \label{eq:db-zero}
  b\circ \de = \de \circ b = 0.
\end{equation}
If we define
\[ \ol{\DR}^{0}(A)\deq \frac{A}{\bb{K} + [A,A]}\isom \frac{\bar{A}}{[A,A]} \]
we can consider the sequence
\begin{equation}
  \label{eq:69}
  \xymatrix{
    0\ar[r] & \ol{\DR}^{0}(A)\ar[r]^-{\de} & \DR^{1}(A) \ar[r]^-{b} &
    [A,A]\ar[r] & 0}.
\end{equation}
By virtue of \eqref{eq:db-zero} this sequence is also a complex; it is called
the \textbf{Quillen complex}.
\begin{lem}
  \label{lem:Afree-Qex}
  Suppose the algebra \(A\) is free. Then the complex \eqref{eq:69} is exact.
\end{lem}
This result is proved for instance in \cite[Lemma 11.5.3]{ginzlect} and has
the following important consequence. Let \(\alpha\in \DR^{1}(A)\) be a 1-form
on an associative affine space, and write \(\alpha = \sum_{i=1}^{k} a_{i}\de
x_{i}\). Clearly
\[ b(\alpha) = \sum_{i=1}^{k} [a_{i},x_{i}] \]
and lemma \ref{lem:Afree-Qex} implies that \(\alpha\) is exact if and only if
the right-hand side vanishes. On the other hand, \(\alpha\) is exact if and
only if there exists a (non-constant) necklace word \(f\in \ol{\DR}^{0}(A)\)
such that \(\alpha = \de f\), in which case, as we saw above, \(a_{i}\) is
just the necklace derivative \(\frac{\partial f}{\partial x_{i}}\). Putting
all together, we obtain:
\begin{teo}
  \label{th:prim}
  Let \(A\) be a free algebra and \(\{g_{1}, \dots, g_{k}\}\) be a subset of a
  set of generators for \(A\). There exist words \(u_{1}, \dots, u_{k}\in A\)
  such that
  \[ \sum_{i=1}^{k} [u_{i},g_{i}] = 0 \]
  if and only if there exists \(f\in \ol{\DR}^{0}(A)\) such that \(u_{i}
  = \frac{\partial f}{\partial x_{i}}\) for every \(i\in \{1,\dots,k\}\).
\end{teo}

\subsection{Relative differential forms}
\label{rel-df}

In order to show other interesting examples of associative varieties it is
necessary to generalize slightly the differential calculus set up in
\S\ref{assoc-df} by developing the notion of differential forms \emph{relative
  to a subalgebra}, which was also introduced by Cuntz and Quillen in
\cite{cq95}.

Let us assume that the associative algebra \(A\) contains a commutative
subalgebra \(B\) that we interpret as an enlarged subspace of ``scalars''.
Then it is natural to require for a derivation defined on \(A\) to vanish not
only on \(\bb{K}\) but on the whole of \(B\); that is, given an \(A\)-bimodule
\(M\) we should consider the set
\[ \Der_{B}(A,M)\deq \set{\theta\in \Der(A,M) | \theta(b)=0 \text{ for all }
  b\in B}. \]
This defines a subfunctor of the functor \eqref{eq:07},
\[ \app{\Der_{B}(A,\mph)}{\bimod{A}}{\mathbf{Set}}. \]
It turns out that this subfunctor also has a universal element: there exists a
pair \(\p{\Omega^{1}(A/B),d}\) consisting of an \(A\)-bimodule
\(\Omega^{1}(A/B)\), whose elements will be called the \textbf{K\"ahler
  differentials of} \(A\) \textbf{relative to} \(B\), and a derivation \(d\in
\Der_{B}(A,\Omega^{1}(A/B))\) vanishing on \(B\), such that for every other
pair \(\p{M,\theta}\) of this kind there exists a unique \(A\)-bimodule
morphism \(\app{f_{\theta}}{\Omega^{1}(A/B)}{M}\) such that \(f_{\theta}\circ
d = \theta\).

Let us consider the tensor algebra over \(A\) of this bimodule,
\[ \Omega^{\bullet}(A/B)\deq \bm{T}_{A}(\Omega^{1}(A/B)). \]
To equip this (graded) algebra with a differential we need again to identify
it with a suitably ``relativized'' version of the universal differential
envelope introduced in \S\ref{univ-dga}. Namely, we consider the category
\(\dgalg{B}\) having as objects the differential graded algebras over the
commutative algebra \(B\) and as arrows the dg-algebra morphisms
\(\app{f}{\p{D,d}}{\p{E,d'}}\) such that \(\rist{f}{B} = \id_{B}\). We have a
functor
\[ \app{(\mph)_{0}}{\dgalg{B}}{\alg{B}} \]
sending a dg-algebra \(\p{D,d}\) over \(B\) to its degree zero part, which is
an algebra over \(B\); the functor \eqref{eq:17} corresponds to the case \(B =
\bb{K}\). It turns out that also in this more general setting the functor
\((\mph)_{0}\) has a left adjoint: for every \(B\)-algebra \(A\) there exists
a pair \(\p{\mathcal{D}(A/B), i}\) consisting of a dg-algebra
\(\mathcal{D}(A/B)\) over \(B\) and a \(B\)-algebra morphism
\(\app{i}{A}{\mathcal{D}(A/B)_{0}}\) such that for every other pair
\(\p{\Gamma,\psi}\) of this kind there exists a unique morphism of dg-algebras
\(\app{u_{\psi}}{\mathcal{D}(A/B)}{\Gamma}\) that makes the diagram
\begin{equation}
  \label{eq:78}
  \xymatrix{
    A\ar[r]^-{i} \ar[dr]_{\psi} & \mathcal{D}(A/B)_{0}\ar[d]^{u_{\psi}}\\
    & \Gamma_{0}}
\end{equation}
commute in \(\alg{B}\). The algebra \(\mathcal{D}(A/B)\) can be defined in a
way that closely parallels the construction of \(\mathcal{D}(A)\),
\[ \mathcal{D}(A/B)\deq \bigoplus_{n\in \bb{N}} \mathcal{D}(A/B)_{n}
\quad\text{ with }\quad \mathcal{D}(A/B)_{n}\deq A\otimes_{B}
\underbrace{\bar{A}\otimes_{B} \dots \otimes_{B} \bar{A}}_{n\text{ times}}, \]
the only difference being that now the tensor products are over \(B\) and
\(\ol{A}\) is defined to be
\begin{equation}
  \label{eq:72}
  \ol{A}\deq A/B.
\end{equation}
In particular, \(\Omega^{1}(A/B)\) is isomorphic to \(A\otimes_{B} A/B\). The
differential is still given by the formula \eqref{eq:20}, and the product by
the formula \eqref{eq:21}. Naturally enough, the pair \(\p{\mathcal{D}(A/B),
  i}\) is called the \textbf{universal differential envelope of} \(A\)
\textbf{relative to} \(B\).

The crucial result is that theorem \ref{th:iso-Om-D} generalizes to the new
setting:
\begin{teo}
  The graded algebras \(\Omega^{\bullet}(A/B)\) and
  \(\mathcal{D}(A/B)\) are isomorphic.
\end{teo}
The proof of this fact can be found in \cite{ginzlect} (Theorem 10.7.1). This
means that we have a differential
\[ \app{\de}{\Omega^{\bullet}(A/B)}{\Omega^{\bullet+1}(A/B)} \]
extending to every degree the universal derivation \(\app{d}{A}{\Omega^{1}(A/B)}\).
Theorem \ref{th:coh-Om} generalizes as follows:
\begin{equation}
  \label{eq:coh-Om-rel}
  H^{k}(\Omega^{\bullet}(A/B)) =
  \begin{cases}
    B &\text{ if } k=0\\
    0 &\text{ otherwise.}
  \end{cases}
\end{equation}
Let us define the \textbf{Karoubi-de Rham complex of} \(A\) \textbf{relative
  to} \(B\) as the graded vector space
\begin{equation}
  \DR^{\bullet}(A/B)\deq \frac{\Omega^{\bullet}(A/B)}
  {\comm{\Omega^{\bullet}(A/B)}{\Omega^{\bullet}(A/B)}}.
\end{equation}
The ``absolute'' Karoubi-de Rham complex of \(A\) then coincides with
\(\DR^{\bullet}(A/\bb{K})\). Also note that \(\DR^{0}(A/B) = A/[A,A]\) does
not actually depend on \(B\), so that the choice of scalars in \(A\) does not
affect the space of regular functions on the associative variety determined by
\(A\).

The differential calculus introduced in \S\ref{assoc-df} readily extends to the
relative case: for every derivation \(\theta\in \Der_{B}(A)\) relative to
\(B\) we have a degree \(-1\) ``interior product''
\[ \app{i_{\theta}}{\DR^{\bullet}(A/B)}{\DR^{\bullet}(A/B)} \]
and a degree \(0\) ``Lie derivative''
\[ \app{\mathcal{L}_{\theta}}{\DR^{\bullet}(A/B)}{\DR^{\bullet}(A/B)} \]
whose concrete expressions are still given by \eqref{eq:33} and \eqref{eq:36},
respectively.

The constructions in \S\ref{quillen} generalize as follows. The recipe
\(u\de v\mapsto [u,v]\) still defines a map \(\Omega^{1}(A/B)\to [A,A]\) that
descends to a map \(\app{b}{\DR^{1}(A/B)}{[A,A]}\). Moreover, if we define
\(\ol{\DR}^{0}(A/B)\) to be the linear space
\[ \ol{\DR}^{0}(A/B)\deq \frac{A}{B + \comm{A}{A}} \]
then we can again write a ``relative'' Quillen complex as follows:
\begin{equation}
  \xymatrix{
    0\ar[r] & \ol{\DR}^{0}(A/B)\ar[r]^-{\de} & \DR^{1}(A/B) \ar[r]^-{b} &
    [A,A]\ar[r] & 0}.
\end{equation}
However, this sequence may no longer be exact even when \(A\) is free.

\subsection{Quiver path algebras}
\label{qpalg}

Another important class of examples of associative varieties arises by
considering \emph{path algebras} of quivers; let us briefly review their
construction.

A \textbf{quiver} is a directed graph with no constraints on the kind
and the number of its edges; in particular it may have loops and/or multiple
edges between the same pair of vertices, as for instance in
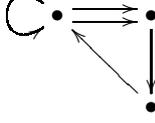
\[ \xymatrix{
  \bullet\ar@(ul,dl)[] \ar@<0.2pc>[r] \ar@<-0.2pc>[r] & \bullet \ar[d]\\
  & \bullet \ar[ul]} \]
It is convenient to identify a quiver \(Q\) with the set of its edges; we
shall denote its set of vertices by \(I_{Q}\), or simply by \(I\) if the
particular quiver we are referring to is clear from the context. Given an edge
\(\xi\) of a quiver \(Q\) we shall denote by \(s(\xi)\) its starting vertex,
or \textbf{source}, and by \(t(\xi)\) its ending vertex, or \textbf{target}.

A \textbf{path} in a quiver \(Q\) is a finite sequence of continuous edges in
\(Q\), or equivalently a word of the form \(\xi_{1}\dots \xi_{\ell}\) for some
\(\ell\in \bb{N}\) where \(\xi_{1}, \dots, \xi_{\ell}\in Q\) and \(s(\xi_{i})
= t(\xi_{i+1})\) for every \(1\leq i\leq \ell-1\). (In keeping with standard
practice, the edges that make up a path are written down going from the right
to the left.) The maps \(s\) and \(t\) may be extended from edges to paths in
the obvious way: \(s(\xi_{1}\dots \xi_{\ell}) = s(\xi_{\ell})\) and
\(t(\xi_{1}\dots \xi_{\ell}) = t(\xi_{1})\).

The \textbf{path algebra} of a quiver \(Q\) over the field \(\bb{K}\), denoted
\(\bb{K}Q\), is the \(\bb{K}\)-linear space generated by all the paths in
\(Q\) equipped with the product defined as follows: given two paths \(p_{1}\)
and \(p_{2}\), their product \(p_{1}\cdot p_{2}\) is the concatenation of the
two words \(p_{1}\) and \(p_{2}\) if \(s(p_{1}) = t(p_{2})\) (that is,
\(p_{2}\) ends at the same vertex where \(p_{1}\) starts) and zero otherwise.
It is clear that \(\bb{K}Q\) is an associative algebra over \(\bb{K}\) which
is not commutative in general.

We shall be concerned only with quivers having a \emph{finite} vertex set, say
\(I = \{1,\dots,m\}\). For every \(i\in I\) we shall denote by \(e_{i}\) the
trivial (length zero) path at the vertex \(i\). Obviously, each \(e_{i}\) is
an idempotent element of \(\bb{K}Q\). Moreover, the set \(\p{e_{1}, \dots,
  e_{m}}\) is a \emph{complete set of mutually orthogonal idempotents} for
\(\bb{K}Q\), in the sense that
\[ e_{i}e_{j} = 0 \text{ when } i\neq j \quad\text{ and }\quad \sum_{i\in I}
e_{i} = 1 \]
where \(1\) is the unit of the path algebra. It follows that, as a vector
space, the path algebra decomposes as a direct sum of the form
\begin{equation}
  \label{eq:02}
  \bb{K}Q = \bigoplus_{i,j\in I} (\bb{K}Q)_{ji}
\end{equation}
where \((\bb{K}Q)_{ji}\deq e_{j}\bb{K}Qe_{i}\) is the linear subspace of
\(\bb{K}Q\) spanned by all the paths \(i\to j\) in \(Q\). 

The decomposition \eqref{eq:02} can be seen equivalently as follows. Denote by
\(B\) the subalgebra of \(\bb{K}Q\) generated by the idempotents
\(\p{e_{i}}_{i\in I}\). This algebra is isomorphic to \(\bb{K}^{m}\), seen as
a (commutative) \(\bb{K}\)-algebra with the product defined componentwise. The
embedding \(B\hookrightarrow \bb{K}Q\) then makes \(\bb{K}Q\) an algebra over
\(B\) (in the sense explained in \S\ref{univ-dga}), and in fact it is easy to
check that
\begin{equation}
  \label{eq:palg-talg}
  \bb{K}Q\isom \bm{T}_{B}(E_{Q})
\end{equation}
where \(E_{Q}\) is the \(B\)-bimodule spanned (as a \(\bb{K}\)-linear space)
by the arrows in \(Q\), with left and right actions defined by
\[ e_{i}\xi =
\begin{cases}
  \xi &\text{ if } t(\xi)=i\\
  0 &\text{ otherwise}
\end{cases} \qquad
\xi e_{i} =
\begin{cases}
  \xi &\text{ if } s(\xi)=i\\
  0 &\text{ otherwise}
\end{cases} \qquad \text{ for every } \xi\in Q,\, i\in I. \]
This \(B\)-bimodule structure is just a compact way to package all the
incidence relations described by the quiver \(Q\).

The associative geometry of quiver path algebras has been studied in
\cite{ginz01,blb02}; let us review the main results of these papers. Fix a
quiver \(Q\) and let \(A\deq \bb{K}Q\). In order to obtain a good theory one
has to work relatively to the subalgebra \(B\subseteq A\) introduced above,
using the relative differential calculus reviewed in \S\ref{rel-df}.
Intuitively, this means that derivations and differential forms defined on
\(A\) must keep the vertices of the quiver fixed.

Let us start, then, by considering the complex \(\Omega^{\bullet}(A/B)\), seen
as the universal differential envelope \(\mathcal{D}(A/B)\). We would like to
find a (linear) basis for the homogeneous component of degree \(n\),
\[ \mathcal{D}(A/B)_{n} = A\otimes_{B} A/B\otimes_{B} \dots \otimes_{B} A/B. \]
An element of this space may be written as
\begin{equation}
  \label{eq:74}
  p_{0}\otimes_{B} \de p_{1}\otimes_{B} \dots \otimes_{B} \de p_{n}
\end{equation}
where \(p_{0}, \dots, p_{n}\in A\) and \(p_{1}, \dots, p_{n}\) are paths of
nonzero length (so that their projection in \(A/B\) is nonzero). Suppose that
the path \(p_{k+1}\) ends at vertex \(i\) (that is, \(e_{i}p_{k+1} =
p_{k+1}\)) and the path \(p_{k}\) starts at vertex \(j\) (\(p_{k}e_{j} =
p_{k}\)); then
\[ \de p_{k}\otimes_{B} \de p_{k+1} = \de p_{k}.e_{j} \otimes_{B} e_{i}.\de
p_{k+1} = \de p_{k}.e_{j}e_{i}\otimes_{B} \de p_{k+1} \]
which is zero unless \(i=j\). Clearly these are the only possible relations
between elements of \(\mathcal{D}(A/B)_{n}\), so as a basis for this space we
can take the set of decomposable tensors of the form \eqref{eq:74} where
\(p_{1}, \dots, p_{n}\) are paths of length \(\geq 1\) and \(s(p_{k}) =
t(p_{k+1})\) for every \(0\leq k\leq n\).

Consider now the relative Karoubi-de Rham complex \(\DR^{\bullet}(A/B)\),
starting as usual from the component of degree zero. Let us call a path
\(\xi_{1}\dots \xi_{n}\in A\) an (oriented) \textbf{cycle} if
\[ t(\xi_{1}) = s(\xi_{n}) \]
that is, the path ends at the same vertex where it begins. A \textbf{necklace
  word} in the path algebra \(A\) is a cycle considered up to cyclic
permutations of its component arrows. As in the free algebra case, it is not
difficult to prove the following result (see \cite[Lemma 3.4]{blb02}):
\begin{lem}
  A basis for the linear space \(\DR^{0}(A)\) is given by the set of necklace
  words in \(A\).
\end{lem}
This gives a description for regular functions on the associative variety
determined by \(\bb{K}Q\) quite analogous to the one obtained in
\S\ref{nc-aff-sp} for the associative affine space.

Now let us consider 1-forms. As we saw above, a basis for \(\Omega^{1}(A/B)\)
is provided by expressions of the form \(p_{0}\de p_{1}\) with \(s(p_{0}) =
t(p_{1})\). If \(s(p_{1})\neq t(p_{0})\) then \(p_{0}\de p_{1} = [p_{0}, \de
p_{1}]\) which is killed by the projection onto \(\DR^{1}(A/B)\), so the paths
\(p_{0}\) and \(p_{1}\) must form a cycle in \(Q\). Then an induction argument
over the length of \(p_{1}\) (see \cite[Lemma 3.5]{blb02}) shows that it
suffices to consider the case where \(p_{1}\) has length 1, which means that
it is an edge of the quiver. Summing up, we have the linear isomorphism
\begin{equation}
  \label{eq:79}
  \DR^{1}(A/B)\isom \bigoplus_{\app{\xi}{i}{j}} A_{ij} \de \xi.
\end{equation}
An interesting consequence of these results is that the necklace derivative
operators introduced in \S\ref{nc-aff-sp} can be generalized to the path
algebra of any quiver. Indeed, using the isomorphism \eqref{eq:79} the
differential of any regular function \(f\in \DR^{0}(A)\) can be written in a
unique way as
\begin{equation}
  \label{eq:80}
  \de f = \sum_{\xi\in Q} p_{\xi} \de\xi
\end{equation}
where each \(p_{\xi}\) is a path which goes in opposite direction with respect
to \(\xi\). It follows that we can define a map
\[ \app{\frac{\partial }{\partial \xi}}{\DR^{0}(A)}{A_{ij}\hookrightarrow A} \]
by sending each \(f\in \DR^{0}(A)\) to the path \(p_{\xi}\). We can then
rewrite formula \eqref{eq:80} as
\[ \de f = \sum_{\xi\in Q} \frac{\partial f}{\partial \xi} \de \xi. \]
This expression neatly generalizes formula \eqref{eq:71}, to which it reduces
when \(Q\) is the quiver with a single vertex and \(n\) loops.

In practice, the action of a necklace derivative \(\frac{\partial }{\partial
  \xi}\) on a necklace word \(f\) is computed in the same manner as in the
free algebra case: for each occurrence of the arrow \(\xi\) in \(f\) we write
down the path obtained by removing that arrow from the necklace (starting from
the arrow immediately after it) and then take the sum of all the resulting
paths. Explicitly, \(\eta_{1}\dots \eta_{\ell}\in A\) is a representative for
\(f\) with \(\eta_{1}, \dots, \eta_{\ell}\in Q\) then
\[ \frac{\partial f}{\partial \xi} = \sum_{k=1}^{\ell} \delta_{\eta_{k}\xi}
\eta_{k+1}\dots \eta_{\ell}\eta_{1}\dots \eta_{k-1}. \]
The cohomology of the Karoubi-de Rham complex of \(A\) can also be explicitly
calculated, as first shown in \cite{ginz01,blb02}.
\begin{teo}
  \label{th:coh-qpalg}
  Let \(A\) be the path algebra of a quiver \(Q\). Then
  \begin{equation}
    \label{eq:24b}
    H^{k}(\DR^{\bullet}(A/B)) =
    \begin{cases}
      B &\text{ for } k=0\\
      0 &\text{ for } k\geq 1.
    \end{cases}
  \end{equation}
\end{teo}
This result shows that if the ``right'' choice for the subalgebra of scalar
functions is made then the associative variety determined by a quiver has the
same cohomology as a contractible space, exactly as it happens for associative
affine spaces (theorem \ref{th:coh-sp-aff}).

\section{Representation spaces}
\label{s:rep}

In this section we review the connection between geometric objects defined on
an associative algebra \(A\) and the corresponding objects defined on the
representation spaces of \(A\), thereby making contact between the associative
and the commutative worlds. Our main references for this part are
\cite{ginzlect}, \cite{lebr08} and \cite{ginz09}.

\subsection{Representation spaces and their quotients}
\label{def-Rep}

From now on we suppose that the associative algebra \(A\) is \emph{finitely
  generated}, that is there exists a natural number \(n\in \bb{N}\) such that
\(A\) may be presented as a quotient
\begin{equation}
  \label{eq:76}
  A\isom \falg{\bb{K}}{x_{1}, \dots, x_{n}}/I
\end{equation}
of the free algebra on \(n\) generators by a two-sided ideal \(I\). This
implies that the dg-algebra \(\Omega^{\bullet}(A)\) is also finitely
generated, and that each homogeneous component \(\Omega^{k}(A)\) is finitely
generated as an \(A\)-bimodule.

For every \(d\in \bb{N}\) a \(d\)\textbf{-dimensional representation} of \(A\)
is a morphism of \(\bb{K}\)-algebras
\[ \app{\rho}{A}{\Mat_{d,d}(\bb{K})}, \]
where \(\Mat_{d,d}(\bb{K})\) denotes the algebra of \(d\times d\) matrices
with entries in \(\bb{K}\). It is natural to interpret such matrices as linear
endomorphisms of \(\bb{K}^{d}\) expressed in the canonical basis; then each
representation \(\rho\) defines a left \(A\)-module structure on
\(\bb{K}^{d}\) by putting
\begin{equation}
  \label{eq:corr-rep-mod}
  a.v\deq \rho(a)v \quad\text{ for every } a\in A,\, v\in \bb{K}^{d}.
\end{equation}
Conversely, suppose \(V\) is a \(d\)-dimensional vector space over \(\bb{K}\)
equipped with the structure of a left \(A\)-module. Then the choice of a basis
\(E\) for \(V\) determines, by the same rule \eqref{eq:corr-rep-mod}, a map
\(\app{\rho}{A}{\Mat_{d,d}(\bb{K})}\). Moreover, for every pair \(a,b\in A\)
one has that
\[ \rho(ab)v = ab.v = a.(b.v) = a.\rho(b)v = \rho(a)\rho(b)v \qquad\text{ for
  every } v\in \bb{K}^{d}. \]
It follows that \(\rho(ab) = \rho(a)\rho(b)\) for every \(a,b\in A\), that is
the map \(\rho\) is a morphism of \(\bb{K}\)-algebras, hence a
\(d\)-dimensional representation of \(A\) in the original sense.

Notice that there is a certain amount of arbitrariness in this correspondence
between representations of \(A\) and left \(A\)-modules, which is given by the
choice of the basis \(E\). Choosing a different basis \(E'\) amounts to the
choice of a different isomorphism \(V\simeq \bb{K}^{d}\), leading to a
different map \(\app{\rho'}{A}{\Mat_{d,d}(\bb{K})}\). The two maps \(\rho\)
and \(\rho'\) are then related by the equality
\begin{equation}
  \label{eq:conj-rep}
  \rho'(a) = g\rho(a)g^{-1} \qquad\text{ for every } a\in A,
\end{equation}
where \(g\in \GL_{d}(\bb{K})\) is the invertible \(d\times d\) matrix which
realizes the change of basis from \(E\) to \(E'\). We say in this case that
the representations \(\rho\) and \(\rho'\) are \textbf{equivalent}. One of the
basic goals in the representation theory of associative algebras is to
classify the equivalence classes of finite-dimensional representations of
\(A\), or (equivalently) the \(A\)-modules of finite dimension\footnote{It is
  important to note that there exist associative algebras having no
  finite-dimensional representations. For such algebras one must necessarily
  consider representations of a more general kind, for example as linear
  operators on an infinite-dimensional Hilbert space.}.

In order to attack this problem from a geometric perspective let us consider
the space of all \(d\)-dimensional representation of the algebra \(A\),
\begin{equation}
  \label{eq:def-rep-var}
  \Rep_{d}^{A}\deq \alg{\bb{K}}(A,\Mat_{d,d}(\bb{K})).
\end{equation}
We shall show, following \cite[Chapter 2]{lebr08}, how this space can be seen
in a natural way as an \emph{affine scheme}. It is convenient to start from
the special case in which the algebra \(A\) is free, say on the \(n\)
generators \(x_{1}, \dots, x_{n}\). Then any \(d\)-dimensional representation
of \(A\) can be specified simply by picking a \(n\)-tuple of \(d\times d\)
matrices \(\p{X_{1}, \dots, X_{n}}\) and declaring that \(X_{i}\) is the image
of the generator \(x_{i}\) for every \(i=1\dots n\). It follows that
\begin{equation}
  \label{eq:Rep-free}
  \Rep_{d}^{A} = \underbrace{\Mat_{d,d}(\bb{K})\oplus \dots \oplus
    \Mat_{d,d}(\bb{K})}_{n\text{ times}} = \Mat_{d,d}(\bb{K})^{\oplus n}
\end{equation}
This is clearly an affine algebraic variety (hence, in particular, a scheme);
the corresponding coordinate algebra, that we are going to denote
\(\mathcal{A}_{n,d}\), is isomorphic to the polynomial ring over \(\bb{K}\)
generated by the \(nd^{2}\) indeterminates \((x_{i,jk})_{i=1\dots n,\,
  j,k=1\dots d}\) representing each entry in a generic \(n\)-tuple of
\(d\times d\) matrices:
\begin{equation}
  \label{eq:gen-tup}
  X_{1} =
  \begin{pmatrix}
    x_{1,11} & \hdots & x_{1,1d}\\
    \vdots & \ddots & \vdots\\
    x_{1,d1} & \hdots & x_{1,dd}
  \end{pmatrix}
  \qquad \dots \qquad
  X_{n} =
  \begin{pmatrix}
    x_{n,11} & \hdots & x_{n,1d}\\
    \vdots & \ddots & \vdots\\
    x_{n,d1} & \hdots & x_{n,dd}
  \end{pmatrix}
\end{equation}
Now let us return to the general case of a finitely generated algebra \(A\),
presented as in \eqref{eq:76}. Then a \(d\)-dimensional representation of
\(A\) may again be specified by a \(n\)-tuple of \(d\times d\) matrices
\(\p{X_{1}, \dots, X_{n}}\) such that the map \(\falg{\bb{K}}{x_{1}, \dots,
  x_{n}}\to \Mat_{d,d}(\bb{K})\) defined by \(x_{i}\mapsto X_{i}\) descends to
the quotient \eqref{eq:76}. But this happens if and only if the matrices
\(\p{X_{1}, \dots, X_{n}}\) satisfy every relation determined by the ideal
\(I\), that is the equation
\begin{equation}
  \label{eq:rel-rep}
  p(X_{1}, \dots, X_{n}) =
  \begin{pmatrix}
    0 & \hdots & 0\\
    \vdots & \ddots & \vdots\\
    0 & \hdots & 0
  \end{pmatrix}
\end{equation}
holds for every \(p\in I\) (seen as a noncommutative polynomial in \(n\)
indeterminates).

On the other hand, for each \(p\in I\) we can interpret the left-hand side of
equation \eqref{eq:rel-rep} as the evaluation of the noncommutative polynomial
\(p\) on the generic \(n\)-tuple of matrices \eqref{eq:gen-tup}. The \(d^{2}\)
entries of this matrix are then polynomials in the indeterminates \(x_{i,jk}\)
that generate the algebra \(\mathcal{A}_{n,d}\). Let us denote by \(J_{A}\)
the ideal of \(\mathcal{A}_{n,d}\) generated by all these polynomials as \(p\)
varies in \(I\). The above argument then shows that the set of
\(d\)-dimensional representations of \(A\) coincides with the zero locus of
the ideal \(J_{A}\) in \(\Mat_{d,d}(\bb{K})^{\oplus n}\); it is thus an affine
scheme (of finite type over \(\bb{K}\)), as claimed.

Notice that, since the ideal \(J_{A}\) defined above is not necessarily
radical, the scheme \(\Rep_{d}^{A}\) is not a variety in general. However,
this will always be the case for the particular examples we shall be
interested in (namely, free algebras and quiver path algebras). For this
reason, in what follows we shall usually avoid the more sophisticated
scheme-theoretic point of view and regard \(\Rep_{d}^{A}\) simply as an affine
algebraic variety.

To translate the notion of equivalent representations in geometric terms we
note that the correspondence \eqref{eq:conj-rep} is naturally interpreted as
the definition of a (left) action of the group \(\GL_{d}(\bb{K})\) on
\(\Rep_{d}^{A}\): given \(g\in \GL_{d}(\bb{K})\) and \(\rho\in \Rep_{d}^{A}\),
the representation \(g.\rho\) is defined by
\begin{equation}
  \label{eq:act-G}
  (g.\rho)(a)\deq g\rho(a) g^{-1} \quad\text{ for every } a\in A.
\end{equation}
In fact the center of \(\GL_{d}(\bb{K})\) acts trivially, so that strictly
speaking the group acting is rather its quotient \(G_{d}\deq
\PGL_{d}(\bb{K})\). Clearly, two representations are equivalent if and only if
they are related by the action of an element \(g\in G_{d}\). We conclude that
the equivalence classes of \(d\)-dimensional representations of \(A\) are in
one to one correspondence with the \emph{orbits} of the group action
\eqref{eq:act-G} on \(\Rep_{d}^{A}\). The fundamental goal becomes then to
describe those orbits.

The modern approach to the study of group actions on affine algebraic
varieties goes under the name of \emph{geometric invariant theory}. It is
obviously impossible for us to do justice to this huge topic here. We direct
the reader to the standard reference \cite{mfk94} for a comprehensive
treatment; see also \cite{brio10} for a more concise introduction. We shall
content ourselves with briefly summarizing some results which shed some light
on the above-mentioned problem.

First of all, we remind the reader about the standard notion of quotient in
the algebro-geometric context. Let \(G\) be an algebraic group acting on an
affine variety \(X\). A \textbf{categorical quotient} for this action (in the
category of affine algebraic varieties) is an affine variety \(\cq{X}{G}\)
together with a morphism \(\app{\pi}{X}{\cq{X}{G}}\) such that:
\begin{enumerate}
\item \(\pi\) is \(G\)-invariant: \(\pi(g.x) = \pi(x)\) for every \(g\in G\)
  and \(x\in X\);
\item \(\pi\) is universal among such morphisms: for every \(G\)-invariant
  morphism \(\app{f}{X}{Z}\) there exists a unique morphism
  \(\app{k}{\cq{X}{G}}{Z}\) such that \(f = k\circ\pi\).
\end{enumerate}
As for any universal construction, if a categorical quotient exists then it is
unique up to a unique isomorphism.

One of the main results of geometric invariant theory is that when \(G\) is a
\emph{reductive} algebraic group\footnote{A linear algebraic group is called
  reductive if it does not contain any closed normal unipotent subgroup. Many
  commonly used groups are reductive, including all semisimple groups and
  general linear groups.} acting on an affine variety \(X\) the categorical
quotient \(\cq{X}{G}\) always exists. This follows from a basic result known
as the \emph{Nagata-Hilbert theorem}:
\begin{teo}
  Let \(G\) be a reductive algebraic group acting on the affine algebraic
  variety \(X\). Then the subalgebra \(\bb{K}[X]^{G}\subseteq \bb{K}[X]\)
  consisting of \(G\)-invariant regular functions on \(X\) is finitely
  generated.
\end{teo}
The construction of the quotient then proceeds in the following way. We choose
a set \(\p{f_{1}, \dots, f_{m}}\) of generators for \(\bb{K}[X]^{G}\) and
consider the morphism \(X\to \bb{K}^{m}\) defined by
\[ x\mapsto \p{f_{1}(x), \dots, f_{m}(x)}. \]
The image \(Y\subseteq \bb{K}^{m}\) of this morphism is closed and independent
of the chosen generating set. The induced surjective morphism \(\app{\pi}{X}
{Y}\) is clearly \(G\)-invariant, and it can be shown that it is also
universal (in the sense of point 2 above). Thus the variety \(Y\) is
isomorphic to the categorical quotient \(\cq{X}{G}\) for the given action%
\footnote{For scheme-theoretically inclined readers it is perhaps easier to
  think about the categorical quotient as the spectrum of the ring of
  invariants \(\bb{K}[X]^{G}\); the morphism \(\pi\) is then obtained from the
  algebra embedding \(\bb{K}[X]^{G}\hookrightarrow \bb{K}[X]\) by duality.}.

It is important to note that the fibers of the quotient map \(\pi\) will
\emph{not} consist, in general, of single orbits. However it can be proved
that each fiber of \(\pi\) contains a unique \emph{closed} orbit, so that the
variety \(\cq{X}{G}\) may be seen as a moduli space for closed \(G\)-orbits.

Let us return now to the specific setting of representation spaces. As the
group \(G_{d}\) is reductive and the variety \(\Rep_{d}^{A}\) is affine, the
previous results can be applied in a straightforward way. Moreover, closed
orbits in \(\Rep_{d}^{A}\) are characterized by the following fundamental
result, due to M. Artin \cite{art69}.
\begin{teo}
  The orbit of a representation \(\rho\) is closed in \(\Rep_{d}^{A}\) if and
  only if the corresponding \(A\)-module is semisimple.
\end{teo}
Putting all together, it follows that for each \(d\in \bb{N}\) there exists an
affine algebraic variety (or scheme)
\begin{equation}
  \label{eq:def-R}
  \mathcal{R}_{d}^{A}\deq \cq{\Rep_{d}^{A}}{G_{d}},
\end{equation}
equipped with a surjective morphism \(\app{\pi}{\Rep_{d}^{A}}
{\mathcal{R}_{d}^{A}}\), that parametrizes equivalence classes of
\emph{semisimple} \(d\)-dimensional representations of \(A\).

Let us remark that, depending on the situation at hand, the categorical
quotient \eqref{eq:def-R} may not be the best choice as a quotient space for
\(\Rep_{d}^{A}\); for instance it can be too small, or too big, or too
singular. It is possible, and often useful, to define a more general class of
quotients by taking a \(G_{d}\)-invariant open subset in \(\Rep_{d}^{A}\)
defined by suitable (semi)stability conditions and looking for a variety which
parametrizes those \(G_{d}\)-orbits which are closed in this subset (see again
\cite{mfk94} for the general theory of these ``GIT quotients''). In the next
section we shall see a particular case of this approach, which takes advantage
of some additional structure on \(\Rep_{d}^{A}\) (namely a symplectic form) in
order to construct smaller and more tractable quotient spaces.

\subsection{The correspondence between the associative and the commutative worlds}
\label{a2c-map}

We shall now explain, following \cite[Section 12]{ginzlect}, how each
associative-geometric object defined on the algebra \(A\) induces a
corresponding \(G_{d}\)-invariant object on the space of \(d\)-dimensional
representations of \(A\), and consequently a geometric object on its quotient
spaces.

Let us start from regular functions. By definition, the space of
representations of \(A\) comes equipped with an \emph{evaluation map}
\[ \Rep_{d}^{A}\times A\to \Mat_{d,d}(\bb{K}) \]
given by \(\p{\rho,a}\mapsto \rho(a)\). Keeping the second argument of this
map fixed we see that every \(a\in A\) determines a \emph{matrix-valued
  function} on \(\Rep_{d}^{A}\), that is a map \(\app{\ind{a}}{\Rep_{d}^{A}}
{\Mat_{d,d}(\bb{K})}\) given by \(\rho\mapsto \rho(a)\) for every \(\rho\in
\Rep_{d}^{A}\). Taking the trace of the resulting matrix we obtain a genuine
function \(\app{\tind{a}}{\Rep_{d}^{A}}{\bb{K}}\). Explicitly,
\[ \tind{a}(\rho) = \tr \ind{a}(\rho) = \tr \rho(a). \]
Since \(a\) can be expressed as a polynomial in some generating set \(\{x_{1},
\dots, x_{n}\}\) for \(A\) we see that \(\tind{a}(\rho)\) can be expressed as a
polynomial in the entries of the matrices \(\rho(x_{1}), \dots, \rho(x_{n})\).
But these entries generate the coordinate algebra of \(\Rep_{d}^{A}\), hence
\(\tind{a}\) is a regular function on \(\Rep_{d}^{A}\). It follows that the
correspondence \(a\mapsto \tind{a}\) defines a map
\begin{equation}
  \label{eq:phi}
  \app{\phi}{A}{\bb{K}[\Rep_{d}^{A}]}.
\end{equation}
Observe that this map is \(\bb{K}\)-linear, since for every \(a,b\in A\) one
has
\[ \wtind{a+b}(\rho) = \tr \rho(a+b) = \tr (\rho(a) + \rho(b)) = \tr
\rho(a) + \tr \rho(b) = \tind{a}(\rho) + \tind{b}(\rho), \]
whence \(\phi(a+b) = \phi(a)+\phi(b)\), and similarly for every \(\lambda\in
\bb{K}\) and \(a\in A\) one has
\[ \wtind{\lambda a}(\rho) = \tr \rho(\lambda a) = \tr (\lambda\rho(a)) =
\lambda \tr \rho(a) = \lambda \tind{a}(\rho), \]
whence \(\phi(\lambda a) = \lambda\phi(a)\).

Moreover, the map \(\phi\) vanishes on the linear subspace \([A,A]\subseteq
A\). To check this it is sufficient to show that \(\tind{c}=0\) for every
\(c\in A\) that can be written as a commutator, say \(c = [a,b]\). But then
\[ \tind{c}(\rho) = \tr \rho(c) = \tr (\rho(a)\rho(b) - \rho(b)\rho(a)) = 0 \]
by the cyclicity of the trace. It follows that the map \eqref{eq:phi} descends
from \(A\) to \(\DR^{0}(A)\).

Finally, we claim that the image of \(\phi\) is contained in the subalgebra of
\(\bb{K}[\Rep_{d}^{A}]\) consisting of \(G_{d}\)-invariant functions. To see
this let us start by noting that \(\DR^{0}(A)\) is generated, as a linear
space, by the necklace words in \(A\). Then it suffices to show that for every
necklace word \(a = a_{1}\dots a_{\ell}\) the regular function \(\tind{a}\) is
constant along \(G_{d}\)-orbits. Let us take \(\rho\in \Rep_{d}^{A}\), \(g\in
G_{d}\) and define \(\rho'\deq g.\rho\); then
\[ \tind{a}(\rho') = \tr \rho'(a_{1})\dots \rho'(a_{\ell}) = \tr
g\rho(a_{1})g^{-1}g\rho(a_{2})g^{-1}\dots g\rho(a_{\ell})g^{-1} = \tr
\rho(a_{1})\dots \rho(a_{\ell}) = \tind{a}(\rho) \]
as claimed. We conclude that there is a well defined linear map
\begin{equation}
  \label{eq:tr-DR0-inv}
  \app{\phi}{\DR^{0}(A)}{\bb{K}[\Rep_{d}^{A}]^{G_{d}}}
\end{equation}
which sends a necklace word \(a_{1}\dots a_{\ell}\) to the \(G_{d}\)-invariant
function
\[ \rho\mapsto \tr \rho(a_{1})\dots \rho(a_{\ell}) \]
on \(\Rep_{d}^{A}\). In this sense each regular function on the associative
variety determined by \(A\) induces a corresponding regular function on each
quotient space \(\mathcal{R}_{d}^{A}\).

It should be stressed that the map \eqref{eq:tr-DR0-inv} is far from being
surjective in general. However, it follows from the \emph{first fundamental
  theorem of matrix invariants} (see e.g. \cite[Theorem 1.6]{lebr08}) that the
image of \(\phi\) generates \(\bb{K}[\Rep_{d}^{A}]^{G_{d}}\) as an algebra.

We can now clear up the mystery regarding the lack of a product between
regular functions on associative varieties. The point is, of course, that
the map \eqref{eq:phi} is not a morphism of algebras: given \(a,b\in A\) we
have that
\[ \wtind{ab}(\rho) = \tr \rho(ab) = \tr \rho(a)\rho(b), \]
which is different from the regular function on \(\Rep_{d}^{A}\) given by
\[ \tind{a}(\rho)\cdot \tind{b}(\rho) = \tr \rho(a)\cdot \tr \rho(b) \]
as the trace of a product is not the product of the traces. Obviously there is
a well-defined product between invariant functions on each representation
space \(\Rep_{d}^{A}\), but this product does not come from the multiplication
map on the algebra \(A\) (nor it is easily expressible in terms of the
latter).

To further elaborate on this point let us consider in detail the case of
associative affine space, \(A = \falg{\bb{K}}{x_{1}, \dots, x_{n}}\). As
explained in the previous subsection, \(\Rep_{d}^{A}\) can be identified with
the linear space \(\Mat_{d,d}(\bb{K})^{\oplus n}\). The map
\eqref{eq:tr-DR0-inv} then sends a generic necklace word \(f = x_{i_{1}}\dots
x_{i_{\ell}}\) (with \(i_{1}, \dots, i_{\ell}\in \{1,\dots,n\}\)) to the
\(G_{d}\)-invariant function
\[ \tind{f}(X_{1}, \dots, X_{n}) = \tr X_{i_{1}}\dots X_{i_{\ell}}. \]
Not every invariant function on \(\Mat_{d,d}(\bb{K})^{\oplus n}\) is of this
form; for instance there is no hope of getting the function \(\tr X_{1}\cdot
\tr X_{2}\) from an element of \(\DR^{0}(A)\). Moreover, even the functions in
the image of \(\phi\) are subject to a certain set of relations depending on
\(d\) (see \cite[Chapter 1]{lebr08}). For instance when \(d=2\) one has the
relation
\begin{multline*}
  \tr X_{1}X_{2}X_{3} + \tr X_{2}X_{1}X_{3} - \tr X_{1} \tr X_{2}X_{3} - \tr
  X_{2} \tr X_{1}X_{3} + {}\\
  {} + \tr X_{1}\tr X_{2} \tr X_{3}- \tr X_{1}X_{2} \tr X_{3} = 0.
\end{multline*}
These relations are also invisible at the level of the linear space
\(\DR^{0}(A)\).

Now let us turn our attention to the vector fields. In order to establish a
correspondence between associative vector fields on \(A\) (that is,
derivations \(A\to A\)) and ordinary (algebraic) vector fields on
\(\Rep_{d}^{A}\) we need a description for the tangent space to a point
\(\rho\in \Rep_{d}^{A}\). In fact it is not difficult to prove (see \cite[\S
12.4]{ginzlect}) that there is an isomorphism
\[ T_{\rho}\Rep_{d}^{A}\isom \Der(A,\Mat_{d,d}(\bb{K})) \]
where the \(A\)-bimodule structure on \(\Mat_{d,d}(\bb{K})\) is given by the
left and right actions defined, respectively, by
\[ a.M\deq \rho(a) M \quad\text{ and }\quad M.b\deq M \rho(b) \]
for every \(a,b\in A\) and \(M\in \Mat_{d,d}(\bb{K})\). More explicitly, a
tangent vector to \(\Rep_{d}^{A}\) at \(\rho\) is specified by a
\(\bb{K}\)-linear map \(\app{\varphi}{A}{\Mat_{d,d}(\bb{K})}\) such that
\[ \varphi(ab) = \varphi(a)\rho(b) + \rho(a)\varphi(b) \qquad\text{ for every
} a,b\in A. \]
Now let \(\theta\) be a derivation \(A\to A\); we wish to define a
corresponding global vector field \(\tind{\theta}\) on \(\Rep_{d}^{A}\).
This means that for every \(\rho\in \Rep_{d}^{A}\) we need to specify a
derivation \(\tind{\theta}_{\rho}\) from \(A\) to \(\Mat_{d,d}(\bb{K})\) (with
the \(A\)-bimodule structure induced by \(\rho\)). We set
\[ \tind{\theta}_{\rho}(a)\deq \rho(\theta(a)) = \wind{\theta(a)}(\rho). \]
The derivation property is easy to check:
\[ \tind{\theta}_{\rho}(ab) = \rho(\theta(ab)) = \rho(\theta(a)b + a\theta(b))
= \rho(\theta(a))\rho(b) + \rho(a)\rho(\theta(b)) =
\tind{\theta}_{\rho}(a)\rho(b) + \rho(a)\tind{\theta}_{\rho}(b). \]
We also have the following important result, whose (non-trivial) proof can be
found in \cite[Proposition 12.4.4]{ginzlect}.
\begin{teo}
  For any finitely generated associative algebra \(A\) the map \(\Der(A)\to
  \Gamma(\mathcal{T}\Rep_{d}^{A})\) defined by \(\theta\mapsto \tind{\theta}\) is a
  morphism of Lie algebras.
\end{teo}
Here \(\mathcal{T}\Rep_{d}^{A}\) denotes the tangent sheaf to
\(\Rep_{d}^{A}\) and \(\Gamma(\mathcal{T}\Rep_{d}^{A})\) its space of
global sections.

It can be worthwhile to see explicitly how this correspondence works in the
case \(A = \falg{\bb{K}}{x_{1}, \dots, x_{n}}\). A derivation on a free
algebra is completely specified by sending each generator \(x_{i}\) to any
chosen element \(f_{i}\in A\). Let us write the derivation defined in this way
as
\[ \theta(x_{1}, \dots, x_{n}) = \p{f_{1}, \dots, f_{n}}. \]
The corresponding vector field on \(\Rep_{d}^{A} = \Mat_{d,d}(\bb{K})^{\oplus
  n}\), then, is simply
\[ \tind{\theta}_{(X_{1}, \dots, X_{n})} = \left( f_{1}(X_{1}, \dots, X_{n}),
  \dots, f_{n}(X_{1}, \dots, X_{n})\right) \]
where we have used the fact that the tangent space to a linear space is
canonically isomorphic to the linear space itself.

Finally let us consider the correspondence between associative \(p\)-forms on
\(A\) (that is, elements of the Karoubi-de Rham complex \(\DR^{p}(A)\)) and
ordinary differential forms on \(\Rep_{d}^{A}\). We shall denote by
\(\Omega^{p}\Rep_{d}^{A}\) the sheaf of (ordinary) differential \(p\)-forms on
the affine variety \(\Rep_{d}^{A}\). Notice that
\(\Omega^{\bullet}\Rep_{d}^{A}\) is a dg-algebra when equipped with the
ordinary exterior differential and exterior product.

We start from K\"ahler differentials. Given \(\alpha\in \Omega^{1}(A)\), say
\(\alpha = a_{0}\de a_{1}\), we consider the \emph{matrix-valued differential
  form} on \(\Rep_{d}^{A}\) whose value at \(\rho\in \Rep_{d}^{A}\) is given
by the matrix product
\[ \ind{\alpha}_{\rho}\deq \ind{a}_{0}(\rho)\cdot \de\ind{a}_{1}(\rho) \]
where \(\de\ind{a}_{1}\) is the differential of the matrix-valued function
\(\app{\ind{a}_{1}}{\Rep_{d}^{A}}{\Mat_{d,d}(\bb{K})}\). (In other words, the
\((i,j)\) entry of the matrix \(\de\ind{a}_{1}(\rho)\) is the differential of
\(\ind{a}_{1}(\rho)_{ij}\), seen as a function of \(\rho\).) Clearly
\(\ind{\alpha}\) is an element of \(\Gamma(\Omega^{1}\Rep_{d}^{A})\otimes
\Mat_{d,d}(\bb{K})\), that is a matrix-valued global section of the sheaf of
1-forms on \(\Rep_{d}^{A}\). Exactly as we did above for regular functions, we
can turn \(\ind{\alpha}\) into a scalar-valued differential form by taking
traces. We denote by \(\tind{\alpha}\) the corresponding global section of
\(\Omega^{1}\Rep_{d}^{A}\):
\[ \tind{\alpha}_{\rho} = \tr \ind{a}_{0}(\rho) \de\ind{a}_{1}(\rho). \]
It is again immediate to check that the correspondence \(\alpha\mapsto
\tind{\alpha}\) defines a linear map \(\Omega^{1}(A)\to
\Gamma(\Omega^{1}\Rep_{d}^{A})\) which vanishes on the linear subspace
\([A,\Omega^{1}(A)]\subseteq \Omega^{1}(A)\) and is constant along
\(G_{d}\)-orbits. This correspondence then induces a map
\[ \DR^{1}(A)\to \Gamma(\Omega^{1}\Rep_{d}^{A})^{G_{d}} \]
which realizes the correspondence between associative 1-forms and
\(G_{d}\)-invariant differential forms on \(\Rep_{d}^{A}\).

A similar procedure works for differential forms of degree \(p>1\): given a
\(p\)-form \(\omega\in \Omega^{p}(A)\), say \(\omega = a_{0}\de a_{1}\dots \de
a_{p}\), there is a corresponding matrix-valued differential \(p\)-form on
\(\Rep^{A}_{d}\) (that is, an element of \(\Gamma(\Omega^{p}\Rep_{d}^{A})
\otimes \Mat_{d,d}(\bb{K})\)) whose value at \(\rho\) is
\[ \ind{\omega}_{\rho}\deq \ind{a}_{0}(\rho)\cdot \de\ind{a}_{1}(\rho)
\wedge \dots \wedge \de\ind{a}_{p}(\rho) \]
where the exterior product \(\wedge\) is extended from 1-forms to \(d\times
d\) matrices of 1-forms in the obvious way\footnote{Namely, \((A\wedge B)_{ij}
  = \sum_{k=1}^{d} A_{ik}\wedge B_{kj}\). Notice that the resulting product is
  no longer skew-symmetric.}.

Taking the trace of the resulting matrix we get the scalar-valued \(p\)-form
\begin{equation}
  \label{eq:corr-df}
  \tind{\omega}_{\rho} = \tr \ind{a}_{0}(\rho)\cdot \de\ind{a}_{1}(\rho)
  \wedge \dots \wedge \de\ind{a}_{p}(\rho).
\end{equation}
The map \(\omega\mapsto \tind{\omega}\) vanishes on the subspace of
\(\Omega^{p}(A)\) spanned by graded commutators and is constant along the
orbits of the group \(G_{d}\). Thus we obtain a map
\[ \DR^{\bullet}(A)\to \Gamma(\Omega^{\bullet}\Rep_{d}^{A})^{G_{d}} \]
sending associative differential forms to \(G_{d}\)-invariant differential
forms on \(\Rep_{d}^{A}\).

For example when \(A = \falg{\bb{K}}{x_{1}, x_{2}}\) the associative 1-form
\(\alpha = x_{2}\de x_{1}\) corresponds to the differential form
\[ \tind{\alpha}_{(X_{1}, X_{2})} = \tr X_{2}\de X_{1} \]
on \(\Mat_{d,d}(\bb{K})\oplus \Mat_{d,d}(\bb{K})\), where \(\de X_{1}\) is the
matrix of differentials \((\de x_{1,ij})_{i,j=1\dots d}\). When \(d=2\) the
corresponding coordinate expression for \(\tind{\alpha}\) is
\[ \tr
\begin{pmatrix}
  x_{2,11} & x_{2,12}\\
  x_{2,21} & x_{2,22}
\end{pmatrix}
\begin{pmatrix}
  \de x_{1,11} & \de x_{1,12}\\
  \de x_{1,21} & \de x_{1,22}
\end{pmatrix}
= x_{2,11}\de x_{1,11} + x_{2,12}\de x_{1,21} + x_{2,21}\de x_{1,12} +
x_{2,22} \de x_{1,22}. \]
Similarly, given the 2-form \(\omega = x_{1}\de x_{2} x_{1}\de x_{1}\) (or
rather its equivalence class in \(\DR^{2}(A)\)) the corresponding 2-form on
\(\Mat_{d,d}(\bb{K})\oplus \Mat_{d,d}(\bb{K})\) reads
\[ \tind{\omega}_{(X_{1}, X_{2})} = \tr (X_{1}\de X_{2}\wedge X_{1}\de X_{1}). \]
When \(d=2\) the corresponding coordinate expression in the basis consisting
of the 2-forms \(\de x_{i,jk}\wedge \de x_{\ell,pq}\) involves 16 terms, and
the count goes up very quickly as \(d\) increases.
Already from these simple examples the convenience in dealing with associative
forms compared to ordinary ones is rather evident.

\subsection{Quiver representation spaces}
\label{qrep}

We now consider in particular the case when \(A\) is the path algebra of a
quiver. In this connection let us note that quivers are a fundamental tool in
the representation theory of associative algebras; we refer the interested
reader to the textbook \cite{ass06} for more information about this topic.

Let \(A = \bb{K}Q\) be the path algebra of a quiver \(Q\) with vertex set \(I
= \{1,\dots,m\}\). Recall the important role played by the
(finite-dimensional, semisimple, commutative) subalgebra \(B\subseteq A\)
spanned by the complete set of idempotents \((e_{i})_{i\in I}\) corresponding
to trivial paths in \(Q\). As we saw in \S\ref{qpalg}, the path algebra \(A\)
can be seen as a tensor algebra over \(B\); it is then natural to consider
only those representations of \(A\) which keep track of this structure.

Observe now that \(B\)-algebra structures on \(\Mat_{d,d}(\bb{K})\) are in one
to one correspondence with direct sum decompositions of the linear space
\(V\deq \bb{K}^{d}\),
\begin{equation}
  \label{eq:dec-Kd}
  V = \bigoplus_{i\in I} V_{i},
\end{equation}
such that \(\sum_{i} \dim V_{i} = d\). Explicitly, one defines a morphism
\(B\to \Mat_{d,d}(\bb{K})\) by sending the idempotent \(e_{i}\) to the matrix
representing the map \(\app{j_{i}\circ \pi_{i}}{V}{V}\), where
\(\app{\pi_{i}}{V}{V_{i}}\) is the canonical projection and
\(\app{j_{i}}{V_{i}}{V}\) is the canonical immersion of the \(i\)-th factor.
As the only invariants of the decomposition \eqref{eq:dec-Kd} are the
dimensions of the subspaces \(V_{i}\), we conclude that each \(B\)-algebra
structure on \(\Mat_{d,d}(\bb{K})\) is completely specified by a vector
\[ \bm{d} = \p{d_{1}, \dots, d_{m}}\in \bb{N}^{m} \]
such that \(d_{1} + \dots + d_{m} = d\). This \(m\)-tuple of natural numbers
is called the \textbf{dimension vector} of the representation.


Now we would like to characterize the space of representations of \(A\) with a
fixed dimension vector \(\bm{d}\). To this end recall the natural isomorphism
\eqref{eq:27} given by the universal property of the tensor algebra. In the
present situation, it can be used to obtain a bijection
\[ B\downarrow \alg{\bb{K}}(\bm{T}_{B}(E_{Q}),\Mat_{d,d}(\bb{K})) \isom
\bimod{B}(E_{Q},U(\Mat_{d,d}(\bb{K}))) \]
between the set of \(B\)-algebra morphisms from \(\bm{T}_{B}(E_{Q})\isom A\)
to \(\Mat_{d,d}(\bb{K})\) (that is, representations of the path algebra which
respect the \(B\)-algebra structure) and the set of \(B\)-bimodule morphisms
\(E_{Q}\to \Mat_{d,d}(\bb{K})\). Such a morphism is completely determined by
sending each arrow \(\app{\xi}{i}{j}\) in \(Q\) to a linear map
\[ \app{\rho(\xi)}{V_{i}}{V_{j}} \]
between the subspaces corresponding to the source and target vertices of
\(\xi\). It follows that the space of representations of \(A\) with dimension
vector \(\bm{d} = \p{d_{1}, \dots, d_{m}}\) coincides with the linear space
\begin{equation}
  \label{eq:Rep-Q}
  \Rep(Q,\bm{d})\deq \bigoplus_{i,j\in I} \bigoplus_{\app{\xi}{i}{j}}
  \Mat_{d_{j},d_{i}}(\bb{K}).
\end{equation}
The notion of equivalence between representations must also be slightly
adjusted, in order to preserve the chosen \(B\)-algebra structure on
\(\Mat_{d,d}(\bb{K})\). Namely, we consider the subgroup of \(\GL(V) =
\GL_{d}(\bb{K})\) consisting of the endomorphisms of \(V\) which preserve the
direct sum decomposition \eqref{eq:dec-Kd}. This means acting on each subspace
\(V_{i}\) with a copy of the general linear group \(\GL(V_{i}) =
\GL_{d_{i}}(\bb{K})\), hence the subgroup in question is
\begin{equation}
  \label{eq:pr-gl}
  \prod_{i\in I} \GL_{d_{i}}(\bb{K}).
\end{equation}
Explicitly, the action of an \(m\)-tuple \(g = \p{g_{1}, \dots, g_{m}}\) on a
point \(\rho\in \Rep(Q,\bm{d})\) is
\[ \rho(\xi)\mapsto g_{j}\rho(\xi)g_{i}^{-1} \quad\text{ for every }
\app{\xi}{i}{j} \text{ in } Q. \]
It is immediate to note that the subgroup \(H\) consisting of
\(m\)-tuples of the form \((\lambda I_{d_{1}}, \dots, \lambda I_{d_{m}})\) for
some \(\lambda\in \bb{K}^{*}\) acts trivially on \(\Rep(Q,\bm{d})\), so that
we can just as well consider the group
\begin{equation}
  \label{eq:gl-d}
  G_{\bm{d}}\deq (\prod_{i\in I} \GL_{d_{i}}(\bb{K}))/H.
\end{equation}
We are now in the same situation already considered in \S\ref{def-Rep}:
namely, we have the action of the linear reductive group \eqref{eq:gl-d} on
the affine algebraic variety \eqref{eq:Rep-Q}. We can thus consider the
corresponding categorical quotient,
\begin{equation}
  \label{eq:naive-quot}
  \cq{\Rep(Q,\bm{d})}{G_{\bm{d}}},
\end{equation}
whose points correspond to closed \(G_{\bm{d}}\)-orbits in \(\Rep(Q,\bm{d})\),
that is equivalence classes of semisimple representations of \(Q\) with
dimension vector \(\bm{d}\). These spaces have been extensively studied in the
literature, starting from the seminal paper \cite{lp90}. 

As remarked at the end of \S\ref{def-Rep}, one can also consider more general
GIT quotients of \(\Rep(Q,\bm{d})\) obtained by imposing suitable
(semi)stability conditions. The reader can hardly do better than consult
\cite{ginz09} for a detailed review of the moduli problem for quiver
representations.

The correspondence between associative and commutative objects described in
\S\ref{a2c-map} generalizes to the above setting as soon as we consider
derivations and differential forms on \(A\) relative to the subalgebra \(B\).
For instance, the condition that a derivation \(\app{\theta}{A}{A}\) vanishes
on \(B\) is exactly what is needed in order to insure that the induced vector
field \(\tind{\theta}\) on \(\Rep(Q,\bm{d})\) preserves the chosen
\(B\)-bimodule structure on \(\Mat_{d,d}(\bb{K})\). As regards differential
forms, the recipe \eqref{eq:corr-df} defines a map
\[ \DR^{\bullet}(A/B)\to \Gamma(\Omega^{\bullet}\Rep(Q,\bm{d}))^{G_{\bm{d}}} \]
relating relative differential forms with ordinary differential forms on
\(\Rep(Q,\bm{d})\) which are invariant with respect to the \(B\)-bimodule
preserving group \(G_{\bm{d}}\subseteq G_{d}\).

\section{Associative symplectic geometry and applications}
\label{s:appl}

In this section we review the idea, introduced by Kontsevich \cite{kont93} and
developed by Ginzburg \cite{ginz01}, of considering the associative analogue
of \emph{symplectic structures}, which play a fundamental role in the
Hamiltonian approach to dynamical systems. Using the differential calculus for
associative algebras reviewed in section \ref{s:ncdf}, every proof from
standard symplectic geometry can be translated verbatim to the new context (at
least insofar it only uses the dg-algebraic properties of the de Rham
complex).

In the second part of the section we briefly survey some applications of the
resulting formalism to the theory of finite-dimensional integrable systems%
\footnote{It must be mentioned that the idea of relating finite-dimensional
  integrable systems to integrable equations on associative algebras goes back
  at least to the pioneering works \cite{os98,ms00}.}. In particular we shall
recover the solution of some models of Calogero-Moser type by the classical
projection method of Olshanetsky and Perelomov.

\subsection{Symplectic structures on associative varieties}
\label{as-symp}

We shall follow the very clear exposition in \cite[Section 14]{ginzlect}. Let
\(A\) be an associative algebra over the field \(\bb{K}\) of characteristic
zero. Given a 2-form \(\omega\in \DR^{2}(A)\) we can define a
\(\bb{K}\)-linear map
\begin{equation}
  \label{eq:57}
  \app{\omega^{\flat}}{\Der(A)}{\DR^{1}(A)}
\end{equation}
by \(\theta\mapsto i_{\theta}(\omega)\). The 2-form \(\omega\) is said to be
\textbf{nondegenerate} if this map is a bijection, in which case we denote its
inverse by \(\app{\omega^{\sharp}}{\DR^{1}(A)}{\Der(A)}\). By definition,
\(\omega^{\sharp}\) maps a 1-form \(\alpha\) to the unique derivation such
that \(i_{\omega^{\sharp}(\alpha)}(\omega) = \alpha\).

An \textbf{associative symplectic variety} is a pair \(\p{A,\omega}\)
consisting of an associative algebra \(A\) and an associative 2-form
\(\omega\) which is closed (\(\de\omega = 0\in \DR^{3}(A)\)) and nondegenerate
in the above sense.

Let \(\p{A,\omega}\) be an associative symplectic variety. A derivation
\(\theta\in \Der(A)\) is called \textbf{symplectic} if
\(\mathcal{L}_{\theta}(\omega) = 0\). We shall denote by \(\Der^{\omega}(A)\)
the linear subspace of \(\Der(A)\) consisting of symplectic derivations. This
is a Lie subalgebra of \(\Der(A)\) since, given two symplectic derivations
\(\theta\) and \(\eta\), we have by \eqref{eq:37a} that
\begin{equation}
  \label{eq:58}
  \mathcal{L}_{\comm{\theta}{\eta}}(\omega) = \comm{\mathcal{L}_{\theta}(\omega)}
  {\mathcal{L}_{\eta}(\omega)} = 0.
\end{equation}
\begin{lem}
  \label{th:symp-der}
  A derivation \(\theta\) is symplectic if and only if \(i_{\theta}(\omega)\)
  is closed in \(\DR^{1}(A)\).
\end{lem}
The standard proof via Cartan's formula \eqref{eq:34} goes through in the
obvious way. It follows that the image of the isomorphism \eqref{eq:57}
restricted to \(\Der^{\omega}(A)\) coincides with the linear subspace of
closed 1-forms in \(\DR^{1}(A)\).

To each regular function \(f\in \DR^{0}(A)\) we can associate the (obviously
closed) 1-form \(\de f\in \DR^{1}(A)\), hence the corresponding
derivation\footnote{Beware: many authors define \(\theta_{f}\) with the
  opposite sign.}
\begin{equation}
  \label{eq:59}
  \theta_{f}\deq -\omega^{\sharp}(\de f)
\end{equation}
is symplectic, and has every right to be called the \textbf{Hamiltonian
  derivation} determined by \(f\). Thus we have defined a \(\bb{K}\)-linear
map
\begin{equation}
  \label{eq:66}
  \app{\theta}{\DR^{0}(A)}{\Der^{\omega}(A)}
\end{equation}
which sends every regular function on the associative symplectic variety
\((A,\omega)\) to the corresponding Hamiltonian derivation.

Returning to the commutative world, it follows from the discussion in section
\ref{s:rep} that for every \(d\in \bb{N}\) the pair consisting of the affine
variety \(\Rep_{d}^{A}\) and the induced 2-form \(\tind{\omega}\in
\Omega^{2}\Rep_{d}^{A}\) qualifies as an (ordinary) symplectic variety.
Moreover, for every \(f\in \DR^{0}(A)\) the derivation \eqref{eq:59} induces
precisely the Hamiltonian vector field on \(\Rep_{d}^{A}\) determined by the
function \(\tind{f}\in\bb{K}[\Rep_{d}^{A}]\). All these geometric objects are
automatically invariant with respect to the action \eqref{eq:act-G} of the
group \(G_{d}\), hence they descend to every quotient of \(\Rep_{d}^{A}\) with
respect to that action. As we shall see later in this section, by working
directly at the associative-geometric level it is possible to treat in a
unified way any family of dynamical systems whose phase space can be obtained
by a quotient process of this kind.

Now let us look for the associative version of the \emph{Poisson bracket}
naturally associated to a symplectic form. Using the above definitions and the
results established in section \ref{s:ncdf} it is a straightforward task to
verify that the following chain of equalities holds for every \(f,g\in
\DR^{0}(A)\):
\begin{equation}
  \label{eq:60}
  \mathcal{L}_{\theta_{f}}(g) = i_{\theta_{f}}(\de g) =
  i_{\theta_{f}}(-i_{\theta_{g}}(\omega)) = i_{\theta_{g}}i_{\theta_{f}}(\omega)
  = -i_{\theta_{g}}(\de f) = -\mathcal{L}_{\theta_{g}}(f),
\end{equation}
where the various Lie derivative and contraction operators involved are seen
as maps on \(\DR^{\bullet}(A)\), as discussed at the end of \S\ref{assoc-df}.
Let us define the \textbf{Poisson bracket} of \(f\) and \(g\), denoted
\(\poibr{f}{g}\), to be the regular function on \(A\) resulting from any of
the expressions in equation \eqref{eq:60}. Equivalently, this defines a
\(\bb{K}\)-bilinear map
\begin{equation}
  \label{eq:68}
  \app{\poibr{\mph}{\mph}}{\DR^{0}(A)\times \DR^{0}(A)}{\DR^{0}(A)}.
\end{equation}
It follows immediately from \eqref{eq:60} that this bracket operation on
\(\DR^{0}(A)\) is skew-symmetric. The easiest way to prove that it also
satisfies the Jacobi identity is to first make the connection with the
commutator bracket on the corresponding symplectic derivations.

Let us start by noting that, quite generally, given \(\gamma,\eta\in \Der(A)\)
and using equation \eqref{eq:37b} we have
\begin{equation}
  \label{eq:54}
  i_{\comm{\gamma}{\eta}} = \mathcal{L}_{\gamma}\circ i_{\eta} -
  i_{\eta}\circ \mathcal{L}_{\gamma} = (\de\circ i_{\gamma} + i_{\gamma}\circ
  \de)\circ i_{\eta} - i_{\eta}\circ (\de\circ i_{\gamma} + i_{\gamma}\circ \de).
\end{equation}
Then by taking \(\gamma = \theta_{f}\), \(\eta = \theta_{g}\) and applying
\(i_{\comm{\theta_{f}}{\theta_{g}}}\) to \(\omega\) we get
\begin{equation}
  \label{eq:61}
  \begin{aligned}
    i_{\comm{\theta_{f}}{\theta_{g}}}(\omega) &= \de
    i_{\theta_{f}}(i_{\theta_{g}}(\omega)) + i_{\theta_{f}}(\de
    i_{\theta_{g}}(\omega)) - i_{\theta_{g}}(\de i_{\theta_{f}}(\omega)) +
    i_{\theta_{g}}(i_{\theta_{f}}(\de \omega))\\
    &= -\de i_{\theta_{f}}(\de g) - i_{\theta_{f}}(\de \de g) +
    i_{\theta_{g}}(\de \de f)\\
    &= -\de \poibr{f}{g}.
  \end{aligned}
\end{equation}
But \(\theta_{\poibr{f}{g}}\) is, by definition, the only derivation such that
\(i_{\theta_{\poibr{f}{g}}}(\omega) = -\de \poibr{f}{g}\), hence
\begin{equation}
  \label{eq:64}
  \theta_{\poibr{f}{g}} = \comm{\theta_{f}}{\theta_{g}}
\end{equation}
and since the bracket \(\comm{\mph}{\mph}\) satisfies the Jacobi identity, the
same holds true for \(\poibr{\mph}{\mph}\). Summing up, we have proved the
following:
\begin{teo}
  The pair \(\p{\DR^{0}(A),\poibr{\mph}{\mph}}\) is a Lie algebra, and the map
  \(f\mapsto \theta_{f}\) is a Lie algebra morphism from it to
  \(\p{\Der^{\omega}(A), \comm{\mph}{\mph}}\).
\end{teo}
We conclude that the space of regular functions on an associative symplectic
variety \(\p{A,\omega}\) is naturally equipped with a Lie algebra structure.

Classically, the Poisson bracket has also the essential feature of being a
\emph{derivation} in both arguments with respect to the associative (and
commutative) product on the coordinate ring of a symplectic variety; in other
words, it determines a \emph{Poisson algebra} structure on that ring. In the
present setting it makes no sense to impose such a condition on the bracket
\eqref{eq:68}, as there is no associative product on \(\DR^{0}(A)\). However,
the induced bracket
\[ \app{\poibr{\mph}{\mph}}{\bb{K}[\Rep_{d}^{A}]^{G_{d}}\times \bb{K}[\Rep_{d}^{A}]^{G_{d}}}
{\bb{K}[\Rep_{d}^{A}]^{G_{d}}} \]
defined on the image of the map \eqref{eq:tr-DR0-inv} (which generates the
algebra of invariants) by
\[ \poibr{\tind{f}}{\tind{g}}\deq \wtind{\poibr{f}{g}} \]
is a genuine Poisson bracket\footnote{The map \eqref{eq:68} is thus an example
  of an ``\(H_{0}\)-Poisson structure'' as introduced by Crawley-Boevey in
  \cite{cb11}, and also comes from a \emph{double Poisson bracket} on \(A\) in
  the sense of Van den Bergh \cite{vdb08}. Lack of space forces us to postpone
  a discussion of these important notions to the second part of these notes
  \cite{nc2}.} on \(\bb{K}[\Rep_{d}^{A}]^{G_{d}}\), and in fact coincides with
the Poisson bracket obtained by inverting the \(G_{d}\)-invariant symplectic
form \(\tind{\omega}\).

To conclude this quick review of associative symplectic geometry let us
display the analogue of the familiar four-terms exact sequence of Lie algebras
associated to a symplectic variety.
\begin{lem}
  \label{lem:exseq}
  The map \eqref{eq:66} fits into the following exact sequence of Lie
  algebras:
  \begin{equation}
    \label{eq:62}
    0\longrightarrow H^{0}(\DR^{\bullet}(A))\longrightarrow \DR^{0}(A)
    \stackrel{\theta}{\longrightarrow} \Der^{\omega}(A)\longrightarrow
    H^{1}(\DR^{\bullet}(A))\longrightarrow 0.
  \end{equation}
\end{lem}
Here the linear spaces \(H^{0}(\DR^{\bullet}(A))\) and
\(H^{1}(\DR^{\bullet}(A))\) are seen as Lie algebras by equipping them with
the zero bracket.

\subsection{Some examples of associative symplectic varieties}

We now review a few examples of symplectic structures on the associative
varieties introduced in section \ref{s:ncdf}. These examples are exactly the
symplectic structures studied by Ginzburg and Bocklandt-Le Bruyn in
\cite{ginz01,blb02}.

Let us start by looking for symplectic structures on associative affine
spaces. Let \(A\) be a free algebra, seen again as the tensor algebra
\(\bm{T}(V^{*})\) of a \(n\)-dimensional vector space \(V^{*}\) with dual
space \(V\). By definition, an associative symplectic structure on \(A\) is
given by a 2-form \(\omega \in \DR^{2}(A)\) which is closed and
nondegenerate. The nature of closed 2-forms on \(A\) is clarified by the
following result.
\begin{teo}
  \label{th:cl2f-aff}
  When the algebra \(A\) is free the subspace of closed forms in
  \(\DR^{2}(A)\) is canonically isomorphic to \([A,A]\).
\end{teo}
\begin{proof}
  By lemma \ref{lem:Afree-Qex} the Quillen complex
  \[ \xymatrix{
    0\ar[r] & \ol{\DR}^{0}(A)\ar[r]^-{\de} & \DR^{1}(A) \ar[r]^-{b} &
    [A,A]\ar[r] & 0} \]
  is exact. By theorem \ref{th:coh-sp-aff}, the following sequence is also
  exact:
  \[ \xymatrix{
    0\ar[r] & \ol{\DR}^{0}(A)\ar[r]^-{\de} & \DR^{1}(A) \ar[r]^-{\de} &
    \DR^{2}(A)_{\text{cl}}\ar[r] & 0} \]
  A trivial diagram chasing argument then gives an isomorphism \([A,A]\to
  \DR^{2}(A)_{\text{cl}}\). 
\end{proof}
More explicitly, if \(\omega\in \DR^{2}(A)_{\text{cl}}\) is a closed 2-form
and \(\alpha\in \DR^{1}(A)\) is a primitive for \(\omega\) (that is,
\(\de\alpha = \omega\)), then the above-mentioned isomorphism sends \(\omega\)
to \(b(\alpha)\in [A,A]\).

Recall from \S\ref{nc-aff-sp} that \(\Der(A)\isom A\otimes V\) and
\(\DR^{1}(A)\isom A\otimes V^{*}\). Let \(\omega_{V}\) be a symplectic
bilinear form on the vector space \(V\); then \(n = 2k\) for some \(k\in
\bb{N}\) and we can find a set of canonical coordinates \(x_{1}, \dots, x_{k},
y_{1}, \dots, y_{k}\in V^{*}\) such that \(\omega_{V} = \sum_{i=1}^{k} \de
y_{i}\wedge \de x_{i}\). Let us consider the associative 2-form on \(A\) given
by the equivalence class in \(\DR^{2}(A)\) of
\begin{equation}
  \label{eq:css-falg}
  \omega\deq \sum_{i=1}^{k} 1\otimes y_{i}\otimes 1\otimes x_{i}\otimes 1\in
  \Omega^{2}(A).
\end{equation}
Notice that \(\omega = \sum_{i=1}^{k} \de y_{i}\, \de x_{i}\), so that
\(\de\omega = 0\). In fact \(\omega\) is the differential of the ``associative
Liouville 1-form'' \(\alpha = \sum_{i=1}^{k} y_{i}\de x_{i}\).

Now let \(\p{e_{1}, \dots, e_{k}, f_{1}, \dots, f_{k}}\) be the (Darboux)
basis for \(V\) dual to the canonical coordinates on \(V^{*}\) considered
above. Then every derivation \(\theta\in A\otimes V\) may be expressed as
\[ \theta = \sum_{j=1}^{k} (a_{j}\otimes e_{j} + b_{j}\otimes f_{j})
\qquad\text{ for some } a_{1},\dots,a_{k},b_{1},\dots,b_{k}\in A. \]
An easy computation shows that
\[ \omega^{\flat}(\theta) = i_{\theta}(\omega) = \sum_{i=1}^{k}
(\theta(y_{i})\de x_{i} - \theta(x_{i})\de y_{i}) = \sum_{i=1}^{k} (b_{i}\de
x_{i} - a_{i}\de y_{i}) = (\id_{A}\otimes \omega_{V}^{\flat})(\theta), \]
where \(\omega_{V}^{\flat}\) is the vector space isomorphism \(V\to V^{*}\)
induced by the symplectic form \(\omega_{V}\). It follows that the map
\(\omega^{\flat}\) is also an isomorphism, and the associative 2-form
\eqref{eq:css-falg} defines an associative symplectic structure on \(A\).
We shall call this 2-form the \textbf{canonical symplectic structure on the
  associative affine space} \(A\) (with respect to the chosen set of generators
\(x_{1}, \dots, x_{k}, y_{1}, \dots, y_{k}\)). Notice that all these 2-forms
are related by (affine) automorphisms of \(A\). Actually it is not hard to
show that these are the \emph{only} possible associative symplectic forms on
\(A\); in particular, odd-dimensional associative affine spaces have no
symplectic forms.

The subspace of symplectic derivations for the canonical symplectic structure
on \(A\) is easy to characterize. For the sake of simplicity we shall consider
only the associative plane, \(A = \falg{\bb{K}}{x,y}\); the adaptation to the
higher-dimensional cases is immediate. The symplectic form \eqref{eq:css-falg}
reads \(\omega = \de x\, \de y\). The derivation defined by \(\theta(x,y) =
\p{f_{1},f_{2}}\) is symplectic if and only if
\[ \mathcal{L}_{\theta}(\omega) = \de f_{1}\, \de y + \de x\, \de f_{2} = 0. \]
Using the isomorphism given by theorem \ref{th:cl2f-aff} this is equivalent to
the condition
\begin{equation}
  \label{eq:symp-cond}
  [f_{1},y] - [f_{2},x] = 0.
\end{equation}
By theorem \ref{th:prim} the above equation is solved by all the pairs of the
form \(\p{\frac{\partial f}{\partial y}, -\frac{\partial f}{\partial x}}\).
Hence the generic symplectic derivation of \(\p{A,\omega}\) is given by
\begin{equation}
  \label{eq:sder-aff}
  \theta(x,y) = \Bigl(\frac{\partial f}{\partial y}, -\frac{\partial
    f}{\partial x}\Bigr)
\end{equation}
for some \(f\in \ol{\DR}^{0}(A)\). This correspondence is bijective and a Lie
algebra isomorphism, as follows from exactness of the sequence \eqref{eq:62}
(where \(H^{0}(\DR^{\bullet}(A)) = \bb{K}\) and \(H^{1}(\DR^{\bullet}(A))=0\)
by virtue of theorem \ref{th:coh-sp-aff}).

The reason for calling the symplectic structure \eqref{eq:css-falg}
``canonical'' becomes clear when we look at the induced symplectic structure
on the space \(\Rep_{d}^{A} = \Mat_{d,d}(\bb{K})^{\oplus n}\). Given \(\rho\in
\Rep_{d}^{A}\), let us define \(X_{i}\deq \rho(x_{i})\) and \(Y_{i}\deq
\rho(y_{i})\). Then
\[ \tind{\omega}_{\rho} = \tr (\de Y_{1}\wedge \de X_{1} + \dots + \de
Y_{k}\wedge \de X_{k}). \]
This can be interpreted as the canonical symplectic form on the cotangent
bundle
\[ T^{*}\Mat_{d,d}(\bb{K})^{\oplus k} \]
if we identify a point \(\p{X_{1}, \dots, X_{k}, Y_{1}, \dots, Y_{k}}\in
\Rep_{d}^{A}\) with the point \(\p{X_{1}, \dots, X_{k}, \zeta_{1}, \dots,
  \zeta_{k}}\in T^{*}\Mat_{d,d}(\bb{K})^{\oplus k}\), where for each
\(i=1\dots k\) the linear functional \(\zeta_{i}\) is defined by
\begin{equation}
  \label{eq:cor-mlf}
  \zeta_{i}(M) = \tr MY_{i} \quad\text{ for every } M\in
  \Mat_{d,d}(\bb{K}).
\end{equation}
A second source of examples comes from a class of associative symplectic
structures on quiver path algebras. In order to describe them it is useful to
define the following ``doubling'' operation. Given a quiver \(Q\) the
\textbf{opposite} of \(Q\), denoted by \(Q^{\op}\), is the quiver with the
same vertices as \(Q\) and, for each arrow \(\app{\xi}{i}{j}\) in \(Q\), an
arrow \(\app{\xi^{*}}{j}{i}\) going in the opposite direction. The
\textbf{double} of \(Q\), denoted \(\ol{Q}\), is the quiver having the same
set of vertices as \(Q\) and the arrows of \(Q\) and \(Q^{\op}\).

Now let \(Q\) be any quiver and denote by \(A\deq \bb{K}\ol{Q}\) the path
algebra of its double. We continue to denote by \(B\) the subalgebra of \(A\)
spanned by the trivial paths. We consider the associative 2-form on \(A\)
given by (the equivalence class in \(\DR^{2}(A/B)\) of)
\begin{equation}
  \label{eq:63}
  \omega_{Q}\deq \sum_{\xi\in Q} \de \xi^{*}\, \de \xi\in \Omega^{2}(A/B).
\end{equation}
Notice that the sum runs over all the arrows in the original quiver \(Q\).
This 2-form is closed, being the differential of \(\alpha_{Q}\deq \sum_{\xi\in
  Q} \xi^{*}\de \xi\). Furthermore, an argument similar to the one used above
for the 2-form \eqref{eq:css-falg} (using the expression of the path algebra
as a tensor algebra of the \(B\)-bimodule \(E_{\ol{Q}}\)) shows that the map
\[ \omega_{Q}^{\flat}(\theta) = \sum_{\xi\in Q} (\theta(\xi^{*})\de \xi -
\theta(\xi)\de \xi^{*}) \]
is invertible, so that \(\omega_{Q}\) is also nondegenerate.

It follows that every quiver \(Q\) gives origin to an associative symplectic
variety \(\p{\bb{K}\ol{Q},\omega_{Q}}\). By analogy with the previous case, we
shall call the 2-form \eqref{eq:63} the \textbf{canonical symplectic form
  associated to the quiver} \(Q\). One reason is that, when \(Q\) is the
quiver with one vertex and \(k\) loops \(x_{1}, \dots, x_{k}\), the 2-form
\(\omega_{Q}\) coincides with the 2-form \eqref{eq:css-falg}. But the main
reason is that the induced symplectic structures on representation spaces may
again be interpreted as canonical symplectic forms on the cotangent bundle
\[ T^{*}\Rep(Q,\bm{d}) \]
where a point \((\rho(\xi), \rho(\xi^{*}))_{\xi\in Q}\in \Rep(\ol{Q},\bm{d})\)
is identified with the point in \(T^{*}\Rep(Q,\bm{d})\) determined on the base
by the matrices \((\rho(\xi))_{\xi\in Q}\) and on the fiber by the linear
functionals corresponding to the matrices \((\rho(\xi^{*}))_{\xi\in Q}\) via
the isomorphism \eqref{eq:cor-mlf}.

From the relative version of the sequence \eqref{eq:62},
\[ 0\longrightarrow H^{0}(\DR^{\bullet}(A/B))\longrightarrow \DR^{0}(A)
\stackrel{\theta}{\longrightarrow} \Der_{B}^{\omega}(A)\longrightarrow
H^{1}(\DR^{\bullet}(A/B))\longrightarrow 0 \]
(where \(\Der_{B}^{\omega}(A)\) denotes the Lie subalgebra of \(\Der_{B}(A)\)
consisting of symplectic derivations) we get, using the description for the
cohomology of the complex \(\DR^{\bullet}(A)\) provided by theorem
\ref{th:coh-qpalg}, the following short exact sequence of Lie algebras:
\begin{equation}
  \xymatrix{
    0\ar[r] & B\ar[r] & \DR^{0}(A)\ar[r]^-{\theta} & \Der_{B}^{\omega}(A)\ar[r] & 0}
\end{equation}
It follows that the Lie algebra \(\Der_{B}^{\omega}(A)\) of symplectic
derivations can be identified with the quotient space \(\ol{\DR}^{0}(A)\). The
generic symplectic derivation of \(A\) can be written as
\begin{equation}
  \label{eq:sder-qpa}
  \theta(\xi_{i},\xi_{i}^{*}) = \Bigl( \frac{\partial f}{\partial
    \xi_{i}^{*}}, -\frac{\partial f}{\partial \xi_{i}}\Bigr)
\end{equation}
where \(f\in \ol{\DR}^{0}(A)\) and the index \(i\) runs over the arrows in the
quiver \(Q\).

\subsection{Free motion on the associative plane and the rational Calogero-Moser system}
\label{cm-sys}

From now on we specialize to the case \(\bb{K} = \bb{C}\), the field of
complex numbers. We are going to describe some examples of dynamics on the
(complex) associative plane \(A = \falg{\bb{C}}{x,y}\) equipped with the
canonical symplectic form \(\omega = \de y\, \de x\), and the corresponding
flows on representation spaces.

Let us start from the simplest possible system, namely the Hamiltonian
describing the free motion on \(\p{A,\omega}\):
\begin{equation}
  \label{eq:ham-fm}
  H = \frac{1}{2} y^{2}.
\end{equation}
Clearly \(\de H = y\,\de y\), so that the symplectic derivation determined by
\(H\) is
\[ \theta_{H}(x,y) = (y,0). \]
This derivation induces an Hamiltonian vector field on each manifold
\(\Rep_{d}^{A}\isom T^{*}\Mat_{d,d}(\bb{C})\) equipped with the canonical
symplectic form \(\tind{\omega}_{\p{X,Y}} = \tr (\de Y\wedge \de X)\). The
resulting flow is rather trivial, and is given by
\begin{equation}
  \label{eq:free-flow}
  \Phi_{t}(X,Y) = (X + tY,Y).
\end{equation}
Things become much more interesting on certain quotient spaces of
\(\Rep_{d}^{A}\) with respect to the natural action of \(G_{d} =
\PGL_{d}(\bb{C})\). Since we want to reduce the symplectic form
\(\tind{\omega}\) along with the manifold, it is natural to consider a
\emph{symplectic reduction} (or Marsden-Weinstein quotient) of the symplectic
vector space \(\p{\Rep_{d}^{A},\tind{\omega}}\) (see for instance \cite{mfk94,
  or04,etin06} and many other sources). The basic ingredient of this process
is the so-called \emph{moment} (or \emph{momentum}) \emph{map} of the
\(G_{d}\)-action \eqref{eq:act-G}, which in our case is the map%
\footnote{Here and in what follows we shall tacitly identify the Lie algebra
  \(\la{sl}_{d}(\bb{C})\) with its dual using the nonsingular bilinear form
  \(\p{X,Y}\mapsto \tr XY\).}
\[ \app{\mu}{\Rep_{d}^{A}}{\la{sl}_{d}(\bb{C})} \quad\text{ defined by }\quad
\mu(X,Y) = [X,Y]. \]
The reduction then proceeds as described by the following:
\begin{teo}
  \label{th:orb-red}
  Let \(\mathcal{O}\) be an adjoint orbit in the Lie algebra
  \(\la{sl}_{d}(\bb{C})\) and consider the categorical quotient
  \begin{equation}
    \label{eq:red-sp}
    M_{\mathcal{O}}\deq \cq{\mu^{-1}(\mathcal{O})}{G_{d}}
  \end{equation}
  with projection map \(\app{\pi}{\mu^{-1}(\mathcal{O})}{M_{\mathcal{O}}}\).
  Suppose that the action of the group \(G_{d}\) on
  \(\mu^{-1}(\mathcal{O})\subseteq \Rep_{d}^{A}\) is free. Then the variety
  \(M_{\mathcal{O}}\) is smooth and there exists a unique symplectic form
  \(\omega_{\mathcal{O}}\) on \(M_{\mathcal{O}}\) such that
  \begin{equation}
    \pi^{*}(\omega_{\mathcal{O}}) = i^{*}(\tind{\omega})
  \end{equation}
  where \(i\) is the canonical immersion
  \(\mu^{-1}(\mathcal{O})\hookrightarrow \Rep_{d}^{A}\).
\end{teo}
We are going to use this result to recover the phase space and the dynamics of
the \emph{rational Calogero-Moser system} \cite{calo71,moser75} from the free
motion on the associative plane. In fact this is precisely the example that
motivated the initial development of associative symplectic geometry by
Ginzburg in \cite{ginz01}. As this particular reduction is explained in a
number of excellent sources \cite{kks78,pere90,etin06}, we shall be quite
brief.

Let us denote by \(\mathcal{O}_{\nu}\) the adjoint orbit in
\(\la{sl}_{d}(\bb{C})\) of the \(d\times d\) matrix
\[ \nu = \imm\tau
\begin{pmatrix}
  0 & 1 & \hdots & 1\\
  1 & 0 & \ddots & 1\\
  \vdots & \ddots & \ddots & \vdots\\
  1 & 1 & \hdots & 0
\end{pmatrix} \]
for some \(\tau\in \bb{C}^{*}\) (the orbits corresponding to different choices
of \(\tau\) are isomorphic; notice that these are precisely the adjoint orbits
of \emph{minimal nonzero dimension} in \(\la{sl}_{d}(\bb{C})\)). It can be
proved (see e.g. \cite[Theorem 1.22]{etin06}) that the action of \(G_{d}\) on
the inverse image \(\mu^{-1}(\mathcal{O}_{\nu})\) is free. We are thus in a
position to apply theorem \ref{th:orb-red}, obtaining a smooth symplectic
variety of dimension \(2d\) that we denote by
\begin{equation}
  \label{eq:def-Cd}
  \mathcal{C}_{d}\deq \cq{\mu^{-1}(\mathcal{O}_{\nu})}{\PGL_{d}(\bb{C})}.
\end{equation}
As explained for instance in \cite[Section 2.7]{etin06} the variety
\(\mathcal{C}_{d}\) is naturally interpreted as the completed phase space of
the (complexified) Calogero-Moser system. In order to recover the usual
interpretation in terms of particles moving on a (complex) line, let us
restrict to the open dense subset \(U\subset \mathcal{C}_{d}\) consisting of
equivalence classes of pairs where the matrix \(X\) is diagonalizable (in
which case it automatically has distinct eigenvalues). Then a point in \(U\)
can be represented by a pair \(\p{X,Y}\) in which the first matrix is
diagonal, say \(X = \diag(q_{1}, \dots, q_{d})\), where all the \(q_{i}\)'s
are distinct. An easy computation then shows that the matrix \(Y\) must have
the form
\[ Y =
\begin{pmatrix}
  p_{1} & \frac{\imm\tau}{q_{1}-q_{2}} & \hdots & \frac{\imm\tau}{q_{1}-q_{d}}\\
  \frac{\imm\tau}{q_{2}-q_{1}} & p_{2} & \ddots & \vdots\\
  \vdots & \ddots & \ddots & \vdots\\
  \frac{\imm\tau}{q_{d}-q_{1}} & \hdots & \hdots & p_{d}
\end{pmatrix} \]
for some complex numbers \(p_{1}, \dots, p_{d}\). The correspondence
\((X,Y)\mapsto (q_{1}, \dots, q_{d}, p_{1}, \dots, p_{d})\) sets up a
bijection between \(U\) and the cotangent bundle to the space
\[ \bb{C}^{(d)}\deq (\bb{C}^{d}\setminus \Delta)/S_{d} \]
of \(d\)-tuples of unordered distinct complex numbers (here \(\Delta\)
denotes the ``big diagonal'' in \(\bb{C}^{d}\), namely the union of all the
hyperplanes \(x_{i}=x_{j}\) for \(i\neq j\in \{1\dots n\}\)). In these
coordinates, the reduced symplectic form on \(\mathcal{C}_{d}\) (restricted to
\(U\)) reads
\[ \omega_{\mathcal{O}_{\nu}} = \sum_{i=1}^{d} \de p_{i}\wedge \de q_{i}. \]
We thus have a symplectic isomorphism \(U\isom T^{*}\bb{C}^{(d)}\), and the
Hamiltonian induced by the necklace word \eqref{eq:ham-fm} becomes, in the new
coordinates,
\[ \tind{H}(q_{i},p_{i}) = \frac{1}{2} \sum_{i=1}^{d} p_{i}^{2} +
\frac{\tau^{2}}{2} \sum_{i\neq j=1}^{d} \frac{1}{(q_{i}-q_{j})^{2}} \]
which is the Hamiltonian of the rational Calogero-Moser system with coupling
constant \(\tau\).

Notice that by construction the variables \((q_{1}, \dots, q_{d})\) can be
identified with the eigenvalues of the matrix \(X\) at each instant of time.
Since \(X\) evolves according to the very simple law \eqref{eq:free-flow}, the
positions of the \(d\) particles at time \(t\) are completely determined by
finding the eigenvalues of the matrix
\[ X(t) = X(0) + tY(0) \]
that is, by finding the \(d\) roots of an \emph{algebraic} equation. 

This method of solving the rational Calogero-Moser system is well known: it
goes back to the seminal papers by Olshanetsky and Perelomov (see \cite{op81}
and references therein), who called it the \emph{projection method}. Their
basic idea is to consider a geodesic motion in some ``big'' Riemannian
symmetric space (for which the solution curves can be explicitly written down
using the exponential map) and then project these curves on some suitable
quotient space in order to reproduce the dynamics of a nonlinear system with a
smaller number of degrees of freedom.

The formalism of associative symplectic geometry sheds a new light on this
classic procedure, seamlessly incorporating it in a much more general
mechanism for producing an infinite family of dynamical systems (one for each
\(d\in \bb{N}\)) starting from a single associative variety equipped with a
symplectic form and a Hamiltonian function.

To show the fruitfulness of this new point of view let us consider a slight
variation of the symplectic quotient \eqref{eq:def-Cd} obtained by replacing
the orbit \(\mathcal{O}_{\nu}\) defined above with an adjoint orbit
\(\mathcal{O}\) of higher dimension in the Lie algebra
\(\la{sl}_{d}(\bb{C})\). In this case there are some additional complications
due to the fact that the action of \(G_{d}\) on the inverse image
\(\mu^{-1}(\mathcal{O})\) is no longer free in general, and the ordinary
theory of Marsden-Weinstein reduction does not apply. However one can resort
to the more general theory of singular symplectic reductions (see e.g.
\cite{or04}), in which case the quotient \eqref{eq:red-sp} exists as a
\emph{stratified symplectic space}. This approach is taken, for instance, in
\cite{hoch08}. It turns out that the reduced dynamics is then confined on a
smooth symplectic stratum inside the singular quotient space. Moreover, there
exists a dense open subset \(U\) and a system of coordinates \(\p{q_{i},
  p_{i}}_{i=1\dots d}\) and \(\p{\lambda_{ij}}_{i\neq j=1\dots d}\) on it such
that the function induced by the Hamiltonian \eqref{eq:ham-fm} reads as
follows:
\[ \tind{H}(q_{i},p_{i},\lambda_{ij}) = \frac{1}{2} \sum_{i=1}^{d} p_{i}^{2} +
\frac{1}{2} \sum_{i\neq j=1}^{d} \frac{\lambda_{ij}\lambda_{ji}}{(q_{i}-q_{j})^{2}} \]
This describes another class of solvable many-body models known as
\emph{rational Calogero-Moser systems with spin}.

\subsection{Other systems obtained from motions on the associative plane}

Let us show some other examples of integrable dynamical systems which may be
obtained by symplectic reduction from a motion defined on the associative
plane. Consider first the ``harmonic oscillator'' Hamiltonian
\[ H = \frac{1}{2}(y^{2} + \omega^{2}x^{2}). \]
where \(\omega\in \bb{C}\) is a constant. By formula \eqref{eq:sder-aff}, the
symplectic derivation determined by \(H\) is
\[ \theta_{H}(x,y) = (y, -\omega^{2}x). \]
On the symplectic vector space \(\p{\Rep_{d}^{A},\tind{\omega}}\) the induced
Hamiltonian function reads
\begin{equation}
  \label{eq:H-osc}
  \tind{H}(X,Y) = \frac{1}{2} \tr (Y^{2} + \omega^{2} X^{2})
\end{equation}
and Hamilton's equations are \(\dot{X} = Y\), \(\dot{Y} = -\omega^{2}X\).
These equations can also be integrated easily; the corresponding flow is given
by
\begin{equation}
  \label{eq:fl-har}
  \Phi_{t}(X,Y) = \p{X\cos (\omega t) + Y \omega^{-1} \sin (\omega t),
    Y\cos(\omega t) - X\omega \sin(\omega t)}.
\end{equation}
Now let us descend to the symplectic quotient \eqref{eq:def-Cd}. Restricting
once again to the dense open subset \(U\isom T^{*}\bb{C}^{(d)}\) with
canonical coordinates \((q_{1}, \dots, q_{d}, p_{1}, \dots, p_{d})\) we see
that the function \eqref{eq:H-osc} becomes
\[ \tind{H}(q_{i},p_{i}) = \frac{1}{2} \sum_{i=1}^{d} p_{i}^{2} + \frac{\tau^{2}}{2}
\sum_{i\neq j=1}^{d} \frac{1}{(q_{i}-q_{j})^{2}} + \frac{\omega^{2}}{2}
\sum_{i=1}^{d} q_{i}^{2}. \]
This is the Hamiltonian of the rational Calogero-Moser system with the
addition of an external harmonic potential. This model is also completely
integrable in the Liouville sense \cite{op81,pere90}. The position of the
particles at time \(t\) are simply the eigenvalues of the first matrix in the
pair \eqref{eq:fl-har},
\[ X(t) = X(0)\cos (\omega t) + Y(0) \omega^{-1} \sin (\omega t). \]
By performing a symplectic reduction along a higher-dimensional adjoint orbit
we can similarly obtain a version of the rational Calogero-Moser systems with
spin variables and an external harmonic potential.

More generally, we could consider a generic Hamiltonian in ``standard'' form
\begin{equation}
  \label{eq:H-mech}
  H = \frac{1}{2}y^{2} + p(x)
\end{equation}
where \(p\) is a polynomial that will play the role of an external potential
for the Calogero-Moser particles after the reduction step. The symplectic
derivation determined by \eqref{eq:H-mech} is
\[ \theta_{H}(x,y) = \p{y, -p'(x)}. \]
Notice that for this particular class of examples the noncommutativity of the
variables \(x\) and \(y\) is totally irrelevant. The induced flow on
representation spaces is then given by the solutions to the matrix
differential equation
\[ \ddot{X} + p'(X) = 0. \]
Of course, the difficulty in this case is to explicitly solve this equation
(which in general amounts to a system of \(d^{2}\) coupled nonlinear ODEs; for
the harmonic potential these equations become linear and decoupled).

The motions determined by non-standard Hamiltonians on \(\p{A,\omega}\) are
also of considerable interest. For instance let us take, following
\cite[Section 2.8]{etin06},
\begin{equation}
  \label{eq:ham-cmhyp}
  H = \frac{1}{2} xyxy.
\end{equation}
In this case \(\de H = yxy\de x + xyx\de y\) and the associated symplectic
derivation is
\[ \theta_{H}(x,y) = \p{xyx, -yxy}. \]
Descending to the symplectic quotient \(\mathcal{C}_{d}\) and restricting to
the usual open subset \(U\isom T^{*}\bb{C}^{(d)}\) we see that the necklace
word \eqref{eq:ham-cmhyp} induces the function
\begin{equation}
  \label{eq:ham-1}
  \tind{H}(q_{i},p_{i}) = \frac{1}{2} \sum_{i=1}^{d} q_{i}^{2}p_{i}^{2} +
  \frac{\tau^{2}}{2} \sum_{i\neq j=1}^{d} \frac{q_{i}q_{j}}{(q_{i}-q_{j})^{2}}.
\end{equation}
By further restricting to the open subset
\[ U'\deq \set{(q_{i},p_{i})\in U | q_{i}>0 \text{ for every } i=1\dots d} \]
and performing the change of variables
\[ \theta_{i}\deq \log q_{i} \quad\text{ and }\quad \ind{p}_{i}\deq q_{i}p_{i}, \]
the function \eqref{eq:ham-1} becomes
\[ \tind{H}(\theta_{i},\ind{p}_{i}) = \frac{1}{2} \sum_{i=1}^{d}
\ind{p}_{i}^{2} + 2\tau^{2} \sum_{i\neq j=1}^{d} \left(\sinh
\frac{\theta_{i} - \theta_{j}}{2}\right)^{-2} \]
which is the Hamiltonian of the \emph{hyperbolic} Calogero-Moser system. With
a similar change of variables the system with trigonometric potential can also
be obtained.

Finally let us note that this mechanism for producing families of solvable
dynamical systems is by no means limited to Hamiltonian evolution equations.
In fact every derivation \(\theta\in \Der(A)\), not necessarily symplectic,
will give rise to a \(\GL_{d}\)-invariant vector field on each representation
space \(\Rep_{d}^{A}\). If we are able to explicitly solve the corresponding
matrix ODEs, thus obtaining an explicit expression for the integral curves of
this vector field, we can again project these solution curves on suitable
lower-dimensional quotients of \(\Rep_{d}^{A}\) (not necessarily obtained by
symplectic reduction) in order to get a solvable system with a smaller number
of degrees of freedom.

A similar process has been used quite effectively in a series of papers by
Calogero and his coworkers (see \cite{bc06} and references therein). Following
the exposition in \cite{bc06}, the idea is to start from a matrix differential
equation of second order
\begin{equation}
  \label{eq:matrix-ode}
  \ddot{X} = F(X,\dot{X})
\end{equation}
whose solutions can be written explicitly. The function \(F\) is assumed to be
\(\PGL_{d}\)-equivariant, that is
\[ gF(U,\dot{U})g^{-1} = F(gUg^{-1}, g\dot{U}g^{-1}) \qquad\text{ for every }
g\in \PGL_{d}(\bb{C}). \]
Each solution of \eqref{eq:matrix-ode} defines a curve \(X = X(t)\); we only
consider those solutions for which the matrix \(X\) is diagonalizable with
distinct eigenvalues at all times. Then we can look once again at the
corresponding motion of the eigenvalues \(\p{q_{1}, \dots, q_{d}}\) of the
matrix \(X\). In general, the evolution equations for these eigenvalues will
not be self-contained. However, in some particular cases the supplementary
unknowns can be consistently expressed, by means of an appropriate
\emph{ansatz}, as functions of the \(q_{i}\)'s and their derivatives. The
resulting system of \(d\) scalar second order ODEs can be interpreted as the
equations of motion for a dynamical system consisting of \(n\) point particles
on the complex line subject to nonlinear interactions of various sorts.

Let us sketch the natural interpretation of this construction from the point
of view of associative geometry. At each instant of time, the pair of matrices
\(\p{X(t),\dot{X}(t)}\) defines a point in \(\Rep_{d}^{A}\). The evolution
equation \eqref{eq:matrix-ode} can then be interpreted as the definition of a
(not necessarily symplectic) derivation on the associative affine plane. The
resulting \(G_{d}\)-invariants flows on \(\Rep_{d}^{A}\) clearly descend to
the categorical quotient \(\cq{\Rep_{d}^{A}}{G_{d}}\); however, they can be
further projected on a \(2d\)-dimensional submanifold \(M\) inside
\(\cq{\Rep_{d}^{A}}{G_{d}}\) by means of the particular \emph{ansatz} which is
used to get rid of the additional unknowns. The resulting flow on \(M\) then
defines the evolution of the reduced dynamical system. We plan to provide some
detailed examples of this reduction process in a future work.

\subsection{Integrable systems related to quiver varieties}
\label{gh-sys}

To conclude let us present a family of dynamical systems whose phase space may
be obtained as a quotient of the representation spaces of a quiver with more
than one vertex, thus leaving the realm of associative affine spaces. As these
systems were introduced by Gibbons and Hermsen in \cite{gh84} we shall call
them \emph{Gibbons-Hermsen systems}.

For every natural number \(r\geq 1\) let \(Q_{r}\) denote the quiver%
\footnote{Readers of \cite{tac15} should note that the quivers described here
  are not the ``zigzag'' quivers; rather, it is the family of quivers briefly
  considered in Appendix B of that paper.} with two vertices \(1\) and \(2\),
a loop \(a\) at \(1\), an arrow \(\app{x}{2}{1}\) and (if \(r>1\)) \(r-1\)
arrows \(\app{y_{2}, \dots, y_{r}}{1}{2}\) (notice that there is no
\(y_{1}\)). Let \(\ol{Q}_{r}\) denote the double of this quiver; it has an
additional loop \(a^{*}\) at \(1\), an arrow \(\app{x^{*}} {1}{2}\) and
\(r-1\) arrows \(\app{y_{2}^{*}, \dots, y_{r}^{*}}{1}{2}\).

The canonical symplectic form determined by the quiver \(Q_{r}\) is
\[ \omega_{r}\deq \de a^{*}\de a + \de x^{*}\de x + \sum_{i=2}^{r} \de
y_{i}^{*} \de y_{i} \]
We are going to consider free motion on \(\p{\bb{C}\ol{Q}_{r},\omega_{r}}\),
described by the Hamiltonian \(H = \frac{1}{2}a^{*2}\).

Let us consider representations of \(\ol{Q}_{r}\) with dimension vector
\((n,1)\). A point in this space is given by \(2+2r\) matrices. We shall
denote by:
\begin{itemize}
\item \(X\) and \(Y\) the \(n\times n\) matrices representing \(a\) and
  \(a^{*}\);
\item \(v_{1}\) the \(n\times 1\) matrix representing \(x\) and \(w_{1}\) the
  \(1\times n\) matrix representing \(x^{*}\);
\item \(w_{2}, \dots, w_{r}\) the \(1\times n\) matrices representing \(y_{2},
  \dots, y_{r}\) and \(v_{2}, \dots, v_{r}\) the \(n\times 1\) matrices
  representing \(y_{2}^{*}, \dots, y_{r}^{*}\).
\end{itemize}
The natural action of the group \(G_{(n,1)}\isom \GL_{n}(\bb{C})\) on this
data is given by
\begin{equation}
  \label{eq:act-gh}
  g.(X,Y,v_{\alpha},w_{\alpha}) = (gXg^{-1}, gYg^{-1}, gv_{\alpha}, w_{\alpha}g^{-1})
\end{equation}
and the flow induced by the Hamiltonian \(\frac{1}{2}a^{*2}\) is
\[ \Phi_{t}(X,Y,v_{\alpha},w_{\alpha}) = \p{X + tY, Y, v_{\alpha}, w_{\alpha}} \]
The moment map relative to the action \eqref{eq:act-gh} is the map
\(\Rep(\ol{Q}_{r},(n,1))\to \la{gl}_{n}(\bb{C})\) defined by
\[ \mu(X,Y,v_{\alpha},w_{\alpha}) = [X,Y] + v_{1}w_{1} - \sum_{i=2}^{r} v_{i}w_{i} \]
In order to recover the phase space of the Gibbons-Hermsen system we must
consider the (trivial) adjoint orbit in \(\la{gl}_{n}(\bb{C})\) given by the
single point \(\tau I\). We trust the reader to verify that the
\(2nr\)-dimensional variety given by symplectic quotient
\[ \mathcal{C}_{n,r}\deq \cq{\mu^{-1}(\tau I))}{\GL_{n}(\bb{C})} \]
coincides with the symplectic quotient construction considered in \cite{gh84}%
\footnote{Explicitly, the matrices \(X\), \(P\), \(F\) and \(E\) used by
  Gibbons and Hermsen correspond, respectively, to \(X\), \(Y\), the \(n\times
  r\) matrix \((-v_{1},v_{2}, \dots, v_{r})\) and the \(r\times n\) matrix
  \((w_{1}, w_{2}, \dots, w_{r})^{t}\).}.

Let us denote by \(U\) the open dense subset of \(\mathcal{C}_{n,r}\)
consisting of equivalence classes of \((2+2r)\)-tuples for which the matrix
\(X\) is diagonalizable with distinct eigenvalues. The points of this subset
may be parametrized by a set of \(2n\) complex numbers \((x_{1}, \dots, x_{n},
y_{1}, \dots, y_{n})\) and \(n\) additional points \((f_{i},e_{i})\in A_{r}\),
where \(A_{r}\) is the algebraic variety defined by
\[ A_{r}\deq \set{(\xi, \eta)\in \Mat_{1,r}(\bb{C})\times \Mat_{r,1}(\bb{C}) |
  \pair{\xi}{\eta} = 1}/\bb{C}^{*} \]
with \(\lambda\in \bb{C}^{*}\) acting as \((\xi,\eta)\mapsto (\lambda\xi,
\lambda^{-1}\eta)\). The reduced symplectic form restricted to \(U\) reads
\[ \omega = \sum_{i=1}^{n} \left( \de y_{i} \wedge \de x_{i} + \tau\,
  \de e_{i}\wedge \de f_{i}\right) \]
and the Hamiltonian \(\tind{H} = \frac{1}{2}\tr Y^{2}\) projects down to
\[ H = \frac{1}{2}\sum_{i=1}^{n} p_{i}^{2} + \frac{\tau^{2}}{2} \sum_{i\neq
  j=1}^{n} \frac{\pair{f_{i}}{e_{j}} \pair{f_{j}}{e_{i}}}{(q_{i} -
  q_{j})^{2}} \]
The resulting dynamics is connected to the one described by the Calogero-Moser
systems with spin, although in this case the number of internal degrees of
freedom is higher.

Let us emphasize that the system considered above is just a single example
involving a particular family of quivers, a particular choice of the dimension
vector for the corresponding representation spaces, and a particular choice of
Hamiltonian function. Clearly, many variations on this theme are possible.
Actually, one could argue that every quiver possess a large family of
``natural'' dynamical systems defined on the corresponding representation
spaces; these systems may frequently be explicitly solvable and/or integrable
in the Liouville sense.

This possibility was considered by Nekrasov in his survey \cite{nek99} on
many-body integrable systems obtained by symplectic reduction. In
\cite[Section 5.3]{nek99} Nekrasov explicitly poses the problem of determining
which dynamical systems of this form are integrable, both in the smooth and in
the holomorphic setting (problems 5.20 and 5.22). To the best of our
knowledge, these problems are still wide open.


\end{document}